\documentclass[journal]{IEEEtran}
\usepackage{graphicx}
\usepackage{mathrsfs}
\usepackage{color}
\usepackage{amsmath,amsfonts,amssymb,amsthm,epsfig,epstopdf,url,array,eufrak}
\usepackage{bm}
\usepackage{subeqnarray}
\usepackage{cases}
\usepackage{verbatim}
\usepackage{cite}
\usepackage{authblk}
\newtheorem{theorem}{Theorem}

\newtheorem{remark}{Remark}

\newtheorem{lemma}{Lemma}
\newtheorem{property}{Property}

\newcommand{\bs}{\boldsymbol}
\makeatletter

\newcommand{\Rmnum}[1]{\rm \expandafter\@slowromancap\romannumeral #1@}
\makeatother

\begin{document}
\title{Inverse Moment Matching Based Analysis of Cooperative HARQ-IR over Time-Correlated Nakagami Fading Channels}
\author{Zheng~Shi,
        Haichuan~Ding,
        Shaodan~Ma,
        Kam-Weng~Tam, and~Su~Pan
%\thanks{Manuscript received January 14, 2015; revised May 29, 2015; accepted July 19, 2015. The associate editor coordinating the review of this paper and approving it for publication was M. Elkashlan.}
\thanks{Zheng Shi, Shaodan Ma and Kam-Weng Tam are with the Department of Electrical and Computer Engineering, University of Macau, Macao (e-mail:shizheng0124@gmail.com, shaodanma@umac.mo, kentam@umac.mo).}
\thanks{Haichuan Ding was with University of Macau, and is now with the Department of Electrical and Computer Engineering, University of Florida, U.S.A. (email: dhcbit@gmail.com).}
\thanks{Su Pan is with Nanjing University of Posts and Telecommunications, China (email: supan@njupt.edu.cn).}
%\thanks{The corresponding author is Shaodan Ma.}
%\thanks{This work was supported by the Research Committee of University of Macau under grants: MYRG078(Y1-L2)-FST12-MSD and MYRG101(Y1-L3)-FST13-MSD.}
}
%\author[1]{Zheng~Shi\thanks{The authors are with the Department of Electrical and Computer Engineering, University of Macau, Macau.},~Haichuan~Ding,~Shaodan~Ma,~Kam-Weng Tam,~Su~Pan}
%\author[2]{Haichuan~Ding}
%\author[1]{Shaodan~Ma}
%\author[1]{Kam-Weng Tam}
%\author[3]{Su~Pan}
%\affil[1]{Faculty of Science and Technology, University of Macau, Macau}
%\affil[2]{Faculty of Science and Technology, University of Macau, Macau}
%\affil[3]{Nanjing University of Posts and Telecommunications, Nanjing, China}
%
%\thanks{Manuscript received April 8, 2014; revised April 8, 2014.}
\maketitle
\begin{abstract}
This paper analyzes the performance of cooperative hybrid automatic repeat request with incremental redundancy (HARQ-IR) and proposes a new approach of outage probability approximation for performance analysis. A general time-correlated Nakagami fading channel covering fast fading and Rayleigh fading as special cases is considered here. An efficient inverse moment matching method is proposed to approximate the outage probability in closed-form. The effect of approximation degree is theoretically analyzed to ease its selection. Moreover, diversity order of cooperative HARQ-IR is analyzed. It is proved that diversity order is irrelevant to the time correlation coefficient $\rho$ as long as $\rho<1$ and full diversity from both spatial and time domains can be achieved by cooperative HARQ-IR under time-correlated fading channels. The accuracy of the analytical results is verified by computer simulations and the results reveal that cooperative HARQ-IR scheme can benefit from high fading order and low channel time correlation. Optimal rate selection to maximize the long term average throughput given a maximum allowable outage probability is finally discussed as one application of the analytical results.

%investigates outage performance of cooperative hybrid automatic repeat request with incremental redundancy (HARQ-IR) over time-correlated Nakagami fading channels. The channel model is general enough to cover most of prior analyses as its special cases, e.g. fast fading channels. However, the presence of time correlation complicates the outage analysis due to the difficulty of tackling the product of multiple correlated random variables (RVs). Fortunately, an inverse moment matching method can be developed to derive its CDF owing to the bounded property of its inverse moments. Moreover, the diversity order of cooperative HARQ-IR system is also studied. It is proved that full diversity can be achieved under time-correlated fading channels. The analytical results are verified through Monte Carlo simulations, and the numerical results reveal that cooperative HARQ-IR system can benefit from high fading order and low time correlation. Finally, an optimal rate selection scheme is designed to maximize the long term average throughput given a maximum allowable outage.
\end{abstract}

% Note that keywords are not normally used for peerreview papers.
\IEEEpeerreviewmaketitle
\section{Introduction}
\IEEEPARstart{I}{n} wireless communications, wireless signals are generally corrupted by noise, interference and channel fading, etc. To boost the performance of wireless communications, a lot of techniques have been proposed in the past few decades. One promising technique proposed lately is cooperative relaying. It exploits spatial diversity to improve system capacity. Another promising technique is hybrid automatic repeat request (HARQ). It combines automatic repeat request (ARQ) and forward error correction (FEC) techniques to combat the detrimental effect of fading and noise \cite{ergen2009mobile}. Essentially, time diversity and coding gain are exploited for performance enhancement. Basically, there are three types of HARQ: Type-I HARQ, HARQ with chase combining (HARQ-CC) \cite{dahlman20134g} and HARQ with incremental redundancy (HARQ-IR) \cite{chen2013survey}. Among them, HARQ-IR can provide superior performance due to the exploration of extra coding gaining through code combining. Clearly, by combining HARQ-IR with cooperative relaying, not only spatial diversity but also time diversity and coding gain can be exploited to boost the communication performance further. Cooperative HARQ-IR thus has attracted considerable research interest recently \cite{maham2012analysis,stanojev2009energy,jinho2013energy,zennaro2011base,chelli2013performance}.

To fully exploit the benefits of cooperative HARQ-IR and provide a theoretical guidance for system design, performance analysis of cooperative HARQ-IR is necessary and meaningful. Some analytical results have been reported in the literature. For instance, an opportunistically relaying HARQ-IR system is investigated and average throughput as well as outage probability are particularly analyzed in \cite{maham2012analysis}. Since it considers a quasi-static Rayleigh fading channel wherein the channel responses corresponding to all HARQ rounds of a single packet are constant, the analysis is only applicable to a low mobility environment. For high mobility environment, communication channels are usually fast fading, i.e., all HARQ rounds experience independent channel realizations. Under such channels, an upper bound of outage probability is derived for a cooperative HARQ-IR system with a single relay using Jensen's inequality in \cite{stanojev2009energy}. Optimal design to maximize the energy efficiency given an outage constraint is then discussed. Similarly, energy efficiency of a cooperative HARQ-IR system with distributed cooperative beamforming (DCB) is analyzed in \cite{jinho2013energy}. With the analytical results, the optimal number of selected relays for DCB under a certain energy efficiency criterion is found. Regarding to a cooperative HARQ-IR enabled uplink cellular system, \cite{zennaro2011base} derives the outage probability using Gaussian approximation, based on which a base station selection scheme is proposed. Considering the limitations of \cite{stanojev2009energy,jinho2013energy,zennaro2011base} where certain approximations are applied in the analysis, \cite{chelli2013performance} derives the exact outage probability of cooperative HARQ-IR in terms of the generalized Fox's H function. It enables further analysis of the average number of transmissions and the long term average throughput (LTAT) in closed-forms.

As aforementioned, most of prior works consider either quasi-static fading channels \cite{maham2012analysis} or fast fading channels \cite{chelli2013performance,stanojev2009energy,jinho2013energy,zennaro2011base}. They are not applicable to time-correlated fading channels which usually occur in low-to-medium mobility environment. Under time-correlated fading channels, performance analysis of cooperative HARQ-IR becomes challenging due to the involvement of the product of multiple correlated random variables (RVs). It is also essentially different from the analysis of HARQ-CC over time-correlated fading channels in \cite{kim2011optimal,jin2011optimal,chaitanya2014adaptive} where a sum of multiple correlated RVs is concerned. To our best knowledge, only few analytical results of HARQ-IR over time-correlated fading channels are available in \cite{shi2015analysis} and \cite{yang2014performance}. Specifically, in \cite{shi2015analysis}, non-cooperative HARQ-IR operating over time-correlated Rayleigh fading channels is analyzed and outage probability is derived in closed-form based on polynomial fitting technique. On the other hand, outage probability of opportunistically relaying HARQ-IR operating over time-correlated Nakagami fading channels is derived based on a Lognormal approximation in \cite{yang2014performance}. Since the Lognormal approximation is developed based on central limit theorem (CLT) which is valid for independent fading channels, the analytical result in \cite{yang2014performance} is not accurate for channels with medium-to-high time correlation.

In this paper, we take a step further to analyze cooperative HARQ-IR operating over general time-correlated Nakagami fading channels. Notice that Nakagami fading is more general than Rayleigh fading and covers Rayleigh fading as a special case with fading order of $1$. Due to the involvement of cooperative relaying and Nakagami distribution, the analytical approach in \cite{shi2015analysis} can not be directly applied here. Since the outage probability can be written as cumulative distribution function (CDF) of a product of multiple correlated RVs, it is essential to determine the CDF of the product of multiple correlated RVs. After proving its inverse moments are bounded, we find that the CDF can be uniquely determined by matching the inverse moments. An efficient inverse moment matching approximation is then proposed to derive the outage probability in closed-form and the effect of approximation degree is theoretically analyzed. It is found that the outage probability can be eventually derived as a weighted sum of the CDFs of Lognormal RVs and the Lognormal approximation in \cite{yang2014performance} in fact is a special case of our analysis with approximation degree of zero. Diversity order of cooperative HARQ-IR is also investigated. It is proved that full diversity can be achieved by cooperative HARQ-IR under time-correlated fading channels, except fully correlated fading channels (i.e., quasi-static fading channels). Our analytical results are then verified through Monte Carlo simulations. It is shown that our analytical approach performs better than that in \cite{yang2014performance}. It is also revealed that low time correlation and high fading order are beneficial to the system performance. Our analytical results can facilitate the system design to achieve various objectives, e.g., the maximization of long term average throughput, the minimization of average number of transmissions and the minimization of outage probability, etc.. Optimal rate selection to maximize the long term average throughput given different outage constraints is finally discussed as an example.

The rest of this paper is organized as follows. In Section \ref{sec_sys_mod}, cooperative HARQ-IR protocol and general time-correlated Nakagami fading channels are introduced. Inverse moment matching method is introduced and the outage probability is derived in Section \ref{sec_out}, while diversity order of cooperative HARQ-IR is investigated in Section \ref{sec:div_ord}. In Section \ref{sec:sim}, the accuracy of our analytical results is verified and optimal rate selection is discussed as an example. Finally, conclusions are drawn in Section \ref{sec_con}.

\section{System Model} \label{sec_sys_mod}
%This paper aims at investigating cooperative HARQ-IR system over time-correlated Nakagami-m fading channels. To begin with this problem, the system model is introduced first.
A cooperative system including one source node, one relay node and one destination node, is considered as shown in Fig. \ref{fig_system_model_relay}. To improve the transmission reliability, HARQ-IR protocol is adopted in each node. Unlike prior analyses, a general time-correlated fading channel which covers the fast fading channel as a special case is considered. In the following, the cooperative HARQ-IR protocol and time-correlated fading channel model are introduced in detail.
%With HARQ-IR protocol, the relay helps the source to deliver the message once the relay successfully decodes the message. Additionally,
\begin{figure*}[!t]
%\normalsize
  \centering
  % Requires \usepackage{graphicx}
  \includegraphics[width=5in]{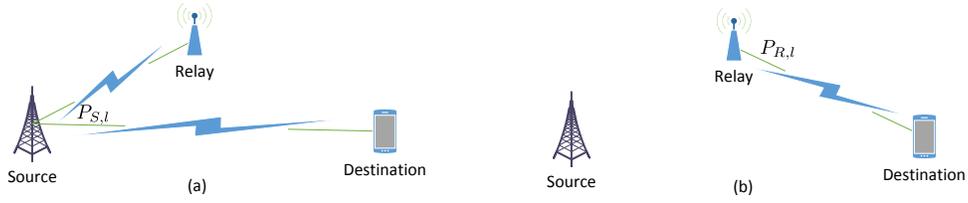}\\
  \caption{A cooperative HARQ-IR system. (a) Broadcasting phase (b) Relaying phase.}\label{fig_system_model_relay}
%\hrulefill
%\vspace*{4pt}
\end{figure*}

\subsection{Cooperative HARQ-IR Protocol}
Following the HARQ-IR protocol, every $b$-bits information message at the source is encoded into a codeword with $M \times L$ symbols \cite{dahlman20134g}, where $M$ is the maximal allowable number of HARQ transmissions. The codeword is then chopped into $M$ sub-codewords, i.e., $\{C_1,\cdots C_l, \cdots C_M\}$, each with length $L$ for transmission in one HARQ round.

%In this paper, we consider HARQ-IR protocol with fixed transmission rate \cite{chelli2013performance,wu2010performance,chelli2014performance}. Thus the codeword is chopped into $M$ sub-codewords of each length $L$, i.e., $\{C_1,\cdots C_l, \cdots C_M\}$. The resulting initial transmission rate is $\mathcal R = b/L$. As a consequence, if $l$ HARQ rounds are utilized to recover $l$ bits information, the effective transmission rate is $\mathcal R/l$.

The cooperative HARQ-IR transmission consists of two phases, i.e., broadcasting and relaying phases, as shown in Fig. \ref{fig_system_model_relay}. During the broadcasting phase, the source sequentially transmits sub-codewords to both the destination and the relay until the maximum number of transmissions is reached or an acknowledgement (ACK) of successful decoding is received from the destination/relay. If the destination successfully decodes the message before the relay, an ACK message will be fed back from the destination to the source and it will be overheard by the relay. The message transmission then completes without moving to the relaying phase. On the contrary, an ACK message will be fed back from the relay and the source will move to the relaying phase. During the relaying phase, the relay encodes the successfully decoded message again and transmits the subsequent sub-codewords to the destination until the maximum number of transmissions is reached or an ACK from the destination is received. The transmitted sub-codewords are different from that transmitted by the source in the broadcasting phase. The destination will use all of the received sub-codewords in both the broadcasting and relaying phases for decoding. In the meanwhile, the source oversees the transmission and listens to the feedback. Once an ACK message is received by the source or the maximum number of transmissions is reached, the transmission for the next $b$ bits information message will be initiated and a new cooperative HARQ-IR transmission will be started. Here an error-free feedback channel is assumed available as \cite{szczecinski2013rate,khosravirad2014rate}, that is, all feedback messages can be correctly decoded.

%the source transmits sub-codewords to the destination. If the destination fails to decode the message, a negative acknowledgement (NACK) message will be fed back. The subsequent sub-codeword will be transmitted in the next HARQ round. Simultaneously, the relay overhears the transmissions and attempts to decode the message using all previously received sub-codewords. Once the relay succeeds to recover the message and the total number of transmissions is still less than $M$, the relaying phase will be initiated. During the relying phase, only the relay is involved into the retransmissions, and the source goes silent. This phase will be sustained until a positive acknowledgement (ACK) message is fed back from the destination or the maximal number of transmissions is reached. It should be mentioned that, as long as the destination successfully decodes the message or the maximum HARQ rounds are attempted, the source will proceed to transmit the next $b$ bits information. Moreover, we assume a error-free feedback channels as \cite{szczecinski2013rate,khosravirad2014rate}, that is, all feedback signals can be correctly decoded.

\subsection{Channel Model}
 In the broadcasting phase, the received signals at the destination and at the relay in the $l$-th HARQ round can be written respectively as
 \begin{equation}
 \label{eqn_received_SD}
 {\bf y}_{SD,l} = h_{SD,l}\sqrt{P_{S,l}}{\bf x}_l + {\bf n}_{SD,l},
\end{equation}
\begin{equation}
\label{eqn_received_SR}
{\bf y}_{SR,l} = h_{SR,l}\sqrt{P_{S,l}}{\bf x}_l + {\bf n}_{SR,l},
\end{equation}
where ${\bf x}_l$ corresponds to the $l$-th sub-codeword $C_l$ and denotes the transmitted signal with unit power in the $l$-th HARQ round; $P_{S,l}$ represents the transmission power in the $l$-th HARQ round; ${\bf n}_{SD,l}$ and ${\bf n}_{SR,l}$ represent zero mean additive white Gaussian noises (AWGNs) with variances $\mathfrak{N}_{SD,l}$ and $\mathfrak{N}_{SR,l}$, respectively; and $h_{SD,l}$ and $h_{SR,l}$ signify the channel coefficients associated with the source-to-destination and the source-to-relay links in the $l$-th HARQ round, respectively.

In the relaying phase, only the relay is involved in the transmission of the subsequent sub-codewords. Accordingly, the received signal at the destination in the $l$-th HARQ round is expressed as
\begin{equation}
\label{eqn_received_RD}
{\bf y}_{RD,l} = h_{RD,l}\sqrt{P_{R,l}}{\bf x}_l + {\bf n}_{RD,l},
\end{equation}
where $P_{R,l}$ represents the transmission power at the relay in the $l$-th HARQ round; ${\bf n}_{RD,l}$ denotes zero mean AWGN with variance $\mathfrak{N}_{RD,l}$; and $h_{RD,l}$ represents the channel coefficient corresponding to the relay-to-destination link.

For notational convenience, we use $h_{ab,l}$ to unify the channel coefficients $h_{SD,l}$, $h_{SR,l}$ and $h_{RD,l}$, where $(a,b) \in \left\{(S,D),(S,R),(R,D)\right\}$. %In this paper, the performance of cooperative HARQ-IR system is investigated in the absence of instantaneous channel state information (CSI) at the transmitter, but only relying on the statistical knowledge of channel information. In addition,
Different from prior analyses \cite{chelli2013performance,jinho2013energy,khosravirad2014rate}, channel time correlation is considered here. Specifically, general time-correlated Nakagami-m fading channels are considered, i.e., the channel coefficients in multiple HARQ rounds $h_{ab,1},\cdots,h_{ab,M}$ are correlated. The amplitudes of the channel coefficients are modeled as multivariate Nakagami-m distributed random variables (RVs) with generalized correlation. The joint probability density function (PDF) corresponding to $|{{{\bf{h}}_{ab}}}| = \left( |h_{ab,1}|,|h_{ab,2}|,\cdots,|h_{ab,M}|\right) $ is given by \cite{beaulieu2011novel}
\begin{multline}\label{eqn_joint_pdf_nakaga}
{f_{|{{\bf{h}}_{ab}}|}}\left( {{x_1}, \cdots ,{x_M}} \right) = \int\nolimits_{t = 0}^\infty  \frac{{{t^{m - 1}}}}{{\Gamma \left( m \right)}}{{\rm{e}}^{ - t}} \\
\times \prod\limits_{l = 1}^M {\frac{{2{x_l}^{2m - 1}}}{{\Gamma \left( m \right){{\left( {\frac{{{\Omega _{ab,l}}\left( {1 - {\lambda _{ab,l}}^2} \right)}}{m}} \right)}^m}}}{e^{ - \frac{{m{x_l}^2}}{{{\Omega _{ab,l}}\left( {1 - {\lambda _{ab,l}}^2} \right)}}}}}  {e^{ - \frac{{{\lambda _{ab,l}}^2t}}{{1 - {\lambda _{ab,l}}^2}}}} \\
 \times
{}_0{F_1}\left( {;m;\frac{{m{x_l}^2{\lambda _{ab,l}}^2t}}{{{\Omega _{ab,l}}{{\left( {1 - {\lambda _{ab,l}}^2} \right)}^2}}}} \right)dt,\,  |{\lambda _{ab,l}}| < 1,
\end{multline}
where $m$ denotes the Nakagami fading order which determines the severity of the fading, $\Omega_{ab,l}$ is the mean channel power gain, i.e., $\Omega_{ab,l}={\rm E}{(\left| {h_{ab,l}} \right|^2)}$, $\lambda_{ab,l}$ denotes generalized correlation coefficient, $\Gamma (\cdot)$ represents Gamma function and ${}_0{F_1}(\cdot)$ denotes the confluent hypergeometric limit function \cite[Eq. 9.14.1]{gradshteyn1965table}. Notice that the time-correlated Rayleigh fading channel in \cite{shi2015analysis} is a special case of this channel model with $m=1$. Moreover, the correlation coefficient $\lambda _{ab,l}$ specifies the cross correlation coefficient ${\rho _{ab}^{l,k}}$ between the squared channel amplitudes ${\left| {{h_{ab,l}}} \right|^2}$ and ${\left| {{h_{ab,k}}} \right|^2}$ as \cite{beaulieu2011novel}
\begin{align}
\label{eqn_cor_fa}
{\rho _{ab}^{l,k}} &= \frac{{{\rm{E}}\left( {{{\left| {{h_{ab,l}}} \right|}^2}{{\left| {{h_{ab,k}}} \right|}^2}} \right) - {\rm{E}}\left( {{{\left| {{h_{ab,l}}} \right|}^2}} \right){\rm{E}}\left( {{{\left| {{h_{ab,k}}} \right|}^2}} \right)}}{{\sqrt {{\rm{Var}}\left( {{{\left| {{h_{ab,l}}} \right|}^2}} \right){\rm{Var}}\left( {{{\left| {{h_{ab,k}}} \right|}^2}} \right)} }} \notag \\
&= {\lambda _{ab,l}}^2{\lambda _{ab,k}}^2,\, 1 \le l \ne k \le M,
\end{align}
where ${\rm Var}(\cdot)$ denotes the operation of variance. It is noteworthy that $|{\lambda _{ab,l}}| < 1$ in (\ref{eqn_joint_pdf_nakaga}), thus $\rho_{ab}^{l,k} < 1$. This channel model covers fast fading channels where the channel coefficients are independent with cross correlation $\rho_{ab}^{l,k} = 0$ as a special case. It is not applicable to quasi-static fading channels where ${h_{ab,1}=h_{ab,2}=\cdots=h_{ab,M}}$ and $\rho_{ab}^{l,k} = 1$. The analysis for quasi-static fading channels has been discussed in \cite{maham2012analysis}. Clearly from (\ref{eqn_joint_pdf_nakaga}), the channel amplitude $|h_{ab,l}|$ follows Nakagami-m distribution (i.e., $|h_{ab,l}| \sim Nakagami(m,\Omega_{ab,l})$) with a PDF of
 \begin{equation}
\label{eqn_nakagami_m}
{f_{\left| {h_{ab,l}} \right|}}\left( x \right) = \frac{{2{{ m }^m}{x^{2m - 1}}}}{{{{\left( {{\Omega _{ab,l}}} \right)}^m}\Gamma \left( m \right)}}\exp \left( { - \frac{m}{{{\Omega_{ab,l}}}}{x^2}} \right),x \in [0, + \infty ).
\end{equation}
%Accordingly, the squared channel amplitude $|h_{ab,l}|^2$ follows Gamma distribution, i.e., $|h_{ab,l}|^2 \sim \mathcal{G}(m,{{{\Omega _{{ab,l}}}/m}})$.
%\begin{equation}\label{eqn_pdf_gamma}
%{f_{{|h_{a,l}|^2}}}\left( x \right) = \frac{{{x^{m - 1}}}}{{{{\left( {\frac{{\Omega _{{a,l}}}}{m}} \right)}^m}\Gamma \left( m \right)}}\exp \left( { - \frac{m}{{\Omega_{{a,l}}}}x} \right),x \in [0, + \infty )
%\end{equation}
%i.e., $|h_{a,l}|^2 \sim \mathcal{G}(m,{{{\Omega _{{a,l}}}/m}})$.
Herein, it should be noted that the channel coefficients associated with different links are independent, that is, ${\bf h}_{SD}$, ${\bf h}_{SR}$ and ${\bf h}_{RD}$ are mutually independent.

Accordingly, the received signal-to-noise ratio (SNR) in the $l$-th HARQ round associated with the link between $a$ and $b$ is written as
\begin{equation}\label{eqn:def_snr_sys_mod}
\gamma_{ab,l} = \frac{P_{a,l}{|h _{ab,l}|^2}}{\mathfrak{N}_{ab,l}},
\end{equation}
and follows Gamma distribution, i.e., $\gamma_{ab,l} \sim \mathcal G(m,\Omega'_{ab,l}/m)$ where $\Omega {'_{ab,l}} = {\Omega _{ab,l}}{P_{a,l}}/{\mathfrak{N}_{ab,l}}$.
By using (\ref{eqn_joint_pdf_nakaga}) and making changes of variables, it is readily proved that the joint distribution of ${\bs \gamma}_{ab} =(\gamma _{ab,1},\cdots,\gamma _{ab,M})$ complies with multivariate Gamma distribution with generalized correlation. More specifically, the joint PDF of ${\bs \gamma}_{ab}$ can be derived as %i.e., ${\bs \gamma}_{ab} \sim {{\mathcal G}} \{( m,\Omega'_{ab,l}, \lambda _{ab,l}), 1 \le l \le M\}$,
\begin{multline}\label{eqn:joint_gamma_l}
{f_{{\bs\gamma _{ab}}}}\left( {{\gamma _1}, \cdots ,{\gamma _M}} \right) = \int\nolimits_{t = 0}^\infty  \frac{{{t^{m - 1}}}}{{\Gamma \left( m \right)}}{{\rm{e}}^{ - t}} \\
\times \prod\limits_{l = 1}^M {\frac{{{\gamma _l}^{m - 1}}}{{\Gamma \left( m \right){{\left( {\frac{{\Omega {'_{ab,l}}\left( {1 - {\lambda _{ab,l}}^2} \right)}}{m}} \right)}^m}}}{e^{ - \frac{{m{\gamma _l}}}{{\Omega {'_{ab,l}}\left( {1 - {\lambda _{ab,l}}^2} \right)}}}}}  {e^{ - \frac{{{\lambda _{ab,l}}^2t}}{{1 - {\lambda _{ab,l}}^2}}}}\\
 \times {}_0{F_1}\left( {;m;\frac{{m{\gamma _l}{\lambda _{ab,l}}^2t}}{{\Omega {'_{ab,l}}{{\left( {1 - {\lambda _{ab,l}}^2} \right)}^2}}}} \right)dt,{\mkern 1mu} |{\lambda _{ab,l}}| < 1.
\end{multline}
Due to the presence of time correlation in the channel coefficients, the analysis becomes much more challenging than the prior works in the literature.

\section{Outage Analysis}\label{sec_out}
The most fundamental performance metric for various HARQ schemes is outage probability. It can well approximate the error probability when Gaussian codes and typical set decoding are applied \cite{caire2001throughput}. In HARQ-IR scheme, message decoding is performed based on the signals received in all the previous HARQ rounds. Outage would happen when the accumulated mutual information per symbol is less than the initial transmission rate $\mathcal R$ \cite{wu2010performance}. Since the cooperative HARQ-IR scheme involves both broadcasting and relaying phases, the destination can acquire information from both the source and the relay. The outage probability at the destination after $K$ HARQ rounds can thus be expressed based on the Total Probability theorem as \cite{chelli2013performance,tabet2007diversity}
%information We define ${P_{out}}\left( K \right)$ as the outage probability at the destination after $K$ HARQ rounds, wherein $K \le M$. In the context of cooperative HARQ schemes, if the earliest HARQ round at which the relay successfully decodes the message is $r$, and only when $r < K$, the relay can help the source to deliver the subsequent packets in the following HARQ rounds. It is worthy to note that if $r \ge K$, %that is, the successful decoding event at the relay occurs after at least $K$ HARQ rounds, so it implies
%it indicates that the relay is unable to assist the source to deliver the message at the first $K$ HARQ rounds. Accordingly, by applying the total probability theorem, ${P_{out}}\left( K \right)$ is formulated as \cite{chelli2013performance,tabet2007diversity}
\begin{equation}\label{eqn_out_k_be_1}
%{P_{out}}\left( K \right) = \sum\nolimits_{r = 1}^{K - 1} {{p_{K,r}}{\left(q_{r-1} - q_r\right)}}  + {p_K}{q_{K-1}},
{P_{out}}\left( K \right) = \sum\nolimits_{r = 1}^{K} {{P_{out}} \left(K|BC=r\right)\Pr \left(BC=r\right)},
\end{equation}
where ${P_{out}} \left(K|BC=r\right)$ denotes the conditional outage probability given there are $r$ broadcasting HARQ rounds among the $K$ HARQ rounds, while $\Pr \left(BC=r\right)$ is the probability that $r$ out of $K$ HARQ rounds are in the broadcasting phase.

%%, but the transmitters (the source and the relay) only have access to the statistical characterization of fading channels
Similarly to \cite{caire2001throughput}, we assume Gaussian codes are applied and channel state information is perfectly known at the receivers. As mentioned in the cooperative protocol, the destination acquires information only from the source in the broadcasting phase, while it gets information only from the relay in the relaying phase. The conditional outage probability ${P_{out}} \left(K|BC=r\right)$ thus can be written as \cite{chelli2013performance,caire2001throughput,tabet2007diversity}
%
%where ${p_{K,r}}$ and ${p_{K}}$ denote the outage probability at the destination after $K$ HARQ rounds given $r < K$ and $r \ge K$, respectively, and $q_r$ is the outage probability at the relay after $r$ HARQ rounds. Therefore, it leads to determining $p_{K,r}$, $p_K$ and $q_r$. From information theoretical perspective, since HARQ-IR accumulates mutual information in all previous HARQ rounds, an outage event of HARQ-IR happens when the accumulated mutual information is less than $\mathcal R$. As derived in \cite{chelli2013performance,tabet2007diversity}, outage probabilities $p_{K,r}$, $p_K$ and $q_r$ can be reformulated as
%
\begin{equation}
{P_{out}}\left( {K|BC = r} \right) = \Pr \left( {{I_{K,r}} \le \mathcal R} \right),
\end{equation}
where ${I_{K,r}}$ is given as \cite{sesia2004incremental},\cite{tse2005fundamentals}
\begin{align}\label{eqn:ins_accumu_inf_2}
&{I_{K,r}} = \frac{1}{L}I (  {{\bf{x}}_1}, \cdots ,{{\bf{x}}_K}; \notag\\
&\qquad\qquad\qquad{{\bf{y}}_{SD,1}}, \cdots , {{\bf{y}}_{SD,r}},{{\bf{y}}_{RD,r + 1}},\cdots , {{\bf{y}}_{RD,K}} |{\bf h}_{K,r} )  \notag\\
&\mathop {\rm{ = }}\limits^{(a)}  \frac{1}{L} \left\{ \sum\limits_{l = 1}^r I\left( { {{\bf{x}}_l};{{\bf{y}}_{SD,l}}|h_{SD,l}} \right)+\sum\limits_{l = r+1}^K I\left( { {{\bf{x}}_l};{{\bf{y}}_{RD,l}}|h_{RD,l}} \right) \right\}.  \notag\\
&= \sum\limits_{l = 1}^r {{{\log }_2}\left( {1 + {\gamma _{SD,l}}} \right){\rm{ + }}} \sum\limits_{l = r + 1}^K {{{\log }_2}\left( {1 + {\gamma _{RD,l}}} \right)},
\end{align}
where $L$ is the number of symbols in each sub-codeword, ${\bf h}_{K,r} =  \{h_{SD,1},\cdots,h_{SD,r},h_{RD,r+1},\cdots,h_{RD,K}\}$, $I\left(\bf{x};\bf{y} | \bf{z} \right)$ denotes the conditional mutual information of random variables $\bf{x}$ and $\bf{y}$ given $\bf{z}$, and (a) holds since the inputs ${{\bf{x}}_1}, \cdots ,{{\bf{x}}_K}$ are independent and the channels are memoryless.
%
%${I_{K,r}} = \frac{1}{L}I\left( {\left. {{{\bf{x}}_1}, \cdots ,{{\bf{x}}_K};{{\bf{y}}_{SD,1}}, \cdots ,{{\bf{y}}_{SD,r}},{{\bf{y}}_{RD,r + 1}}, \cdots ,{{\bf{y}}_{RD,K}}} \right|{\bf h}_{K,r}} \right)$ where $L$ is the number of symbols in each sub-codeword, ${\bf h}_{K,r} =  \{h_{SD,1},\cdots,h_{SD,r},h_{RD,r+1},\cdots,h_{RD,K}\}$, and $I\left(\bf{x};\bf{y} | \bf{z} \right)$ denotes the conditional mutual information of random variables $\bf{x}$ and $\bf{y}$ given $\bf{z}$ \cite{sesia2004incremental},\cite{tse2005fundamentals}. Since memoryless channels are considered and Gaussian codes are applied, i.e., the inputs ${{{\bf{x}}_1}, \cdots ,{{\bf{x}}_K}}$ are independent, the accumulated mutual information per symbol ${I_{K,r}}$ can be derived as \cite[Eq. B.48]{tse2005fundamentals}
%
It follows the conditional outage probability as
\begin{multline}\label{eqn_out_def_re}
{P_{out}}\left( {K|BC = r} \right) = \\
\left\{ {\begin{array}{*{20}{c}}
{\Pr \left( \begin{array}{l}
\sum\limits_{l = 1}^r {{{\log }_2}\left( {1 + {\gamma _{SD,l}}} \right){\rm{ + }}} \\
\sum\limits_{l = r + 1}^K {{{\log }_2}\left( {1 + {\gamma _{RD,l}}} \right)}  < {\cal R}
\end{array} \right)}&{r < K}\\
{\Pr \left( {\sum\limits_{l = 1}^K {{{\log }_2}\left( {1 + {\gamma _{SD,l}}} \right)}  < {\cal R}} \right)}&{r = K.}
\end{array}} \right.
%
%{p_{K,r}} &= \Pr \left( {\sum\limits_{l = 1}^r {{{\log }_2}\left( {1 + {\gamma _{SD,l}}} \right)} {\rm{ + }}\sum\limits_{l = r + 1}^K {{{\log }_2}\left( {1 + {\gamma _{RD,l}}} \right)}  < \mathcal R} \right) \notag\\
% &= \Pr \left( { {Y_{K,r}^D \triangleq \prod\limits_{l = 1}^r {\left( {1 + {\gamma _{SD,l}}} \right)} \prod\limits_{l = r + 1}^K {\left( {1 + {\gamma _{RD,l}}} \right)} } < 2^{\mathcal R}} \right) = {F_{Y_{K,r}^D}}\left( {{2^{\mathcal R}}} \right),
\end{multline}
After simple manipulation, it can be rewritten as
\begin{multline}\label{eqn_out_def_re_1}
{P_{out}}\left( {K|BC = r} \right) = \\
 \left\{ {\begin{array}{*{20}{c}}
{\Pr \left( \begin{array}{l}
Y_{K,r}^D \buildrel \Delta \over = \prod\limits_{l = 1}^r {\left( {1 + {\gamma _{SD,l}}} \right)\times} \\
  \prod\limits_{l = r + 1}^K {\left( {1 + {\gamma _{RD,l}}} \right)}  < {2^{\cal R}}
\end{array} \right)}&{r < K}\\
{\Pr \left( {Y_K^D \buildrel \Delta \over = \prod\limits_{l = 1}^K {\left( {1 + {\gamma _{SD,l}}} \right)}  < {2^{\cal R}}} \right)}&{r = K.}
\end{array}} \right.
\end{multline}
On the other hand, since the relaying phase starts only when the relay can successfully decode the message, i.e., the accumulated mutual information per symbol at the relay is not less than the transmission rate $\mathcal R$, the probability $\Pr \left(BC=r\right)$ can similarly be written as
\begin{multline}\label{eqn:out_rel_def}
\Pr \left( {BC = r} \right) = \\
\left\{ {\begin{array}{*{20}{c}}
{\Pr \left( \begin{array}{l}
\sum\limits_{l = 1}^{r - 1} {{{\log }_2}\left( {1 + {\gamma _{SR,l}}} \right)}  < {\cal R},\\
\sum\limits_{l = 1}^r {{{\log }_2}\left( {1 + {\gamma _{SR,l}}} \right)}  \ge {\cal R}
\end{array} \right)}&{{K > r \ge 1}}\\
{\Pr \left( {\sum\limits_{l = 1}^{K - 1} {{{\log }_2}\left( {1 + {\gamma _{SR,l}}} \right)}  < {\cal R}} \right)}&{{K = r \ge 1}.}
\end{array}} \right.
\end{multline}
By defining ${Y_r^R \triangleq \prod\nolimits_{l = 1}^{r} {\left( {1 + {\gamma _{SR,l}}} \right)} }$ and $Y_0^R \triangleq 0$, it can be rewritten as
\begin{multline}
\label{eqn:out_rel_def_1}
\Pr \left(BC=r\right) = \\
\left\{ \begin{array}{lcl}
{\Pr \left( { Y_{r-1}^R <  2^{\mathcal R}} \right) - \Pr \left( { Y_r^R  <  2^{\mathcal R}} \right) } & {K > r \ge 1} \\
{\Pr \left( { Y_{K-1}^R <  2^{\mathcal R}} \right)} & {K=r \ge 1}.
\end{array}\right.
\end{multline}

Clearly from (\ref{eqn_out_k_be_1}), (\ref{eqn_out_def_re_1}) and (\ref{eqn:out_rel_def_1}), the CDFs of the products of multiple shifted SNRs, i.e., $Y_{K,r}^D$, $Y_K^D$, and $Y_r^R $, are essential for the outage analysis. In the literature, there are two kinds of approaches to derive the CDF of the product of multiple RVs: Mellin transform \cite{chelli2013performance} and moment matching method \cite{provost2005moment,provost2012orthogonal}. Mellin transform is effective for the case with independent RVs and however is inapplicable to our analysis since the multiple SNRs are correlated due to the channel time correlation \cite{poularikas2010transforms}. On the other hand, although the moments of the products of multiple shifted SNRs (i.e., $Y_{K,r}^D$, $Y_K^D$, and $Y_r^R $) exist, their moment generation functions (MGFs) do not exist when the number of transmissions $K$ or $r$ is larger than one as proved in Appendix \ref{app:proof_resu_1}. According to \cite[pp. 176-177]{cramer1999mathematical}, the uniqueness of the CDFs thus can not be guaranteed by matching the moments of $Y_{K,r}^D$, $Y_K^D$, and $Y_r^R $ when $K,r>1$, which would result in notable degradation on the accuracy of outage analysis based on moment matching method.

Fortunately, after analyzing the products of multiple shifted SNRs, we found that their inverse moments do have special properties which can facilitate the derivation and guarantee the uniqueness of their CDFs. Based on these findings, we will propose an effective outage analysis approach based on inverse moment matching method. In the following, the details of inverse moment matching method will be first introduced by taking the analysis of the CDF of $Y_K^D$ as an example. The derivations of the CDFs of $Y_r^R $ and $Y_{K,r}^D$ will be briefly introduced later.

\subsection{Inverse Moment Matching Method}\label{sec:dis_type_I}
Denote the PDF of the product of multiple shifted SNRs corresponding to the source-to-destination link $Y_K^D$ as ${f_{Y_{K}^D}}\left( y \right)$. The inverse moment of $Y_K^D$ is defined as $\alpha_n = {\rm E}\{(Y_{K}^D)^{-n}\} = \int_{ 0 }^\infty  {{y^{ - n}}{f_{Y_{K}^D}}\left( y \right)dy}$. As shown in Appendix \ref{app_inv_mom}, it can be explicitly derived as
\begin{multline}\label{eqn_inv_mom_def}
\alpha_n \approx \sum\limits_{{p_1}, \cdots ,{p_K} \in \left[ {1,{N_Q}} \right]} \frac{{\prod\limits_{l = 1}^K {{w_{{p_l}}}{{\left( {1 + \frac{{\Omega {'_{SD,l}}\left( {1 - {\lambda _{SD,l}}^2} \right){\zeta _{{p_l}}}}}{m}} \right)}^{ - n}}} }}{{\Gamma \left( m \right){{\left( {1 + \sum\limits_{l = 1}^K {\frac{{{\lambda _{SD,l}}^2}}{{1 - {\lambda _{SD,l}}^2}}} } \right)}^m}}}  \\
  \times \Psi _2^{\left( K \right)}\left( {m;m, \cdots ,m;{\varpi _1}{\zeta _{{p_1}}}, \cdots ,{\varpi _K}{\zeta _{{p_K}}}} \right),
\end{multline}
where ${\varpi _l} = {\left( {1 + \sum\limits_{k = 1}^K {\frac{{{\lambda _{SD,k}}^2}}{{1 - {\lambda _{SD,k}}^2}}} } \right)^{ - 1}}\frac{{{\lambda _{SD,l}}^2}}{{1 - {\lambda _{SD,l}}^2}}$, $N_Q$ is the quadrature order, the weights $w_{p_l}$ and abscissas ${{\zeta _{p_l}}}$ for $N_Q$ up to $32$ are tabulated in \cite{rabinowitz1959tables}, and $\Psi _2^{\left( K \right)}(;;)$ denotes the confluent form of Lauricella hypergeometric function \cite[Definition A.20]{mathai2009h}. Moreover, the inverse moment $\alpha_n$ has the following property.
\begin{property}\label{inverse_moment}
The inverse moment $\alpha_n$ is bounded in $(0,1]$ and decreases with $n$. The series $\sum\nolimits_{n=0}^\infty{\alpha_n s^n/n!}$ absolutely converges for some $s > 0$.
\end{property}
\begin{proof}
According to the definition of $Y_K^D$, it is clear that $Y_{K}^D \ge 1$ and the inverse moment $\alpha_n$ decreases to zero with $n$, i.e.,
\begin{equation}\label{eqn:inverse_mom_finite}
0 < \alpha_n < \alpha_{n-1} < \cdots < \alpha_0 = 1.
\end{equation}
It follows that
\begin{equation}\label{eqn:conver_inverse_expla}
\sum\nolimits_{n=0}^\infty{|\alpha_n s^n/n!|} \le \sum\nolimits_{n=0}^\infty{|s|^n/n!} = e ^{|s|},
\end{equation}
which means that the series $\sum\nolimits_{n=0}^\infty{\alpha_n s^n/n!}$ absolutely converges for some $s > 0$ according to Lebesgue's monotone convergence theorem.
\end{proof}
Meanwhile, we have the following lemma about inverse moments from \cite{govindarajulu1962theory}.
\begin{lemma}\label{lem:serie_con_unique}
 \cite[Result 4.14]{govindarajulu1962theory} For any RV $Y$ with CDF of $F_{Y}(y)$, if its inverse moments $\alpha_0= 1, \alpha_1,\cdots,$ are finite and the series $\sum\nolimits_{n=0}^\infty{\alpha_n s^n/n!}$ is absolutely convergent for some $s > 0$, its CDF $F_{Y}(y)$ is the only CDF having $\alpha_0, \alpha_1,\cdots,$ as its inverse moments.
 \end{lemma}

Based on Property \ref{inverse_moment} and Lemma \ref{lem:serie_con_unique}, it can be concluded that the CDF of the product of multiple shifted SNRs corresponding to the source-destination link $Y_K^D$ can be uniquely determined by matching its inverse moments $\alpha_n$. The PDF $f_{Y_K^D}(y)$ can thus be uniquely determined as shown in the following theorem.
\begin{theorem}\label{the:inverse_mom_based1}
By matching the inverse moments $\alpha_n$, the PDF $f_{Y_K^D}(y)$ can be uniquely expressed as
\begin{equation}\label{eqn:ivner_bsed_fun}
{f_{Y_K^D}}\left( y \right) = {f_b}\left( y \right)\sum\limits_{l = 0}^\infty  {\xi _l {y^{ - l}}},
\end{equation}
where ${f_b}\left( y \right)$ is a nontrivial function of $y$ and denotes a base density function with inverse moments ${\nu _l} = \int\nolimits_{-\infty}^\infty  {{y^{ - l}}{f_b}\left( y \right)dy} $ existing for $l=0,1,\cdots$, and the coefficients $\xi _l$ match the inverse moments $\alpha_n$ such that
\begin{equation}\label{eqn:xi_the}
{\alpha _n} = \int\nolimits_{-\infty}^\infty{y^{-n}f_{Y_K^D}(y)dy}= \sum\limits_{l = 0}^\infty  {{\xi _l}{\nu _{n + l}}}  ,\quad n = 0,1, \cdots.
\end{equation}
\end{theorem}
\begin{proof}
Please see Appendix \ref{app:inverse_mom_based}.
\end{proof}

With the unique expression of the PDF $f_{Y_K^D}(y)$ in (\ref{eqn:ivner_bsed_fun}), the PDF can be approximated by truncating the series in (\ref{eqn:ivner_bsed_fun}) as
\begin{equation}\label{eqn:ivner_bsed_fun_trun}
f_{Y_K^D}(y) \approx {\tilde f_N}\left( y \right) = {f_b}\left( y \right)\sum\limits_{l = 0}^N  {\xi _{N,l} {y^{ - l}}} = {f_b}\left( y \right) {{\bs{\xi }}_N}^{\rm T} {{\bf{y}}_N},
\end{equation}
where $N$ denotes the approximation degree, ${{\bf{y}}_N} = {[ {\begin{array}{*{20}{c}}
1&y^{-1}& \cdots &{{y^{-N}}}
\end{array}} ]^{\rm{T}}}$, and ${\bs \xi _N} = {[ {\begin{array}{*{20}{c}}
{\xi _{N,0}}& \cdots &{\xi _{N,N}}
\end{array}} ]^{\rm{T}}}$ is determined by matching the first $N+1$ inverse moments, i.e., ${\alpha _0},\cdots, {\alpha _N}$. Specifically, by matching the first $N+1$ inverse moments as (\ref{eqn:xi_the}), the coefficient vector ${\bs \xi _N}$ should satisfy
\begin{equation}\label{eqn:coeff_matrx_form}
{{\bf{A}}_N}{{\bs{\xi }}_N} = {{\bs{\alpha }}_N},
\end{equation}
where
\begin{equation}\label{eqn:matrix_U_mom_b}
{{\bf{A}}_N} = \left[ {\begin{array}{*{20}{c}}
{{\nu _0}}&{{\nu _1}}& \cdots &{{\nu _N}}\\
{{\nu _1}}&{{\nu _2}}& \cdots &{{\nu _{N + 1}}}\\
 \vdots & \vdots & \ddots & \vdots \\
{{\nu _N}}&{{\nu _{N + 1}}}& \cdots &{{\nu _{2N}}}
\end{array}} \right],
\end{equation}
and ${{\bs{\alpha }}_N} = {[ {\begin{array}{*{20}{c}}
1&{{\alpha}_1}&{{\alpha}_2}& \cdots &{{\alpha}_N}
\end{array}} ]^{\rm T}}$.
For an arbitrary vector ${\bf d }= [d_0,\cdots,d_N]^{\rm T}$, we have ${\bf d}^{\rm{T}}{{\bf{A}}_N} {\bf d} = \int\nolimits_{ 0 }^\infty  {{f_b}\left( y \right){{\left( {\sum\limits_{l = 0}^N {{d _l}{y^{-l}}} } \right)}^2}dy}  \ge 0$ where the equality holds if and only if ${\bf d}=\bf 0$. Therefore, the matrix ${\bf A}_N$ is positive definite and invertible. From (\ref{eqn:coeff_matrx_form}), it follows that
\begin{equation}\label{eqn:direc_inverse_matrix}
%{\tilde f_N}\left( y \right) = {f_b}\left( y \right) \left( {\bf A}_N^{-1}\bs{\alpha}_N \right)^{\rm T} {{\bf{y}}_N}.
{{\bs{\xi }}_N} = {\bf A}_N^{-1}\bs{\alpha}_N.
\end{equation}
Clearly from (\ref{eqn:direc_inverse_matrix}), matrix inversion is involved in the calculation of the coefficient vector in the approximated PDF. It has high complexity and also causes difficulty in the selection of approximation degree. To avoid that, the approximated PDF is reformulated as shown in the following theorem.
\begin{theorem} \label{theorem:recur_coeff}
The approximated PDF can be reformulated as
\begin{equation}\label{eqn:app_re_mom_fina}
{\tilde f_N}\left( y \right) = {f_b}(y)\sum\limits_{l = 0}^N {\eta_l {\bf c}_l^{\rm T} {\bf y}_l},
\end{equation}
where
\begin{equation}\label{eta}
\eta _l = {{{\bf{c}}_l}^{\rm{T}}{{\bs{\alpha }}_l}},
\end{equation} and
\begin{equation}\label{eqn:def_c_lk}
{{\bf{c}}_l} = {\left[ {\begin{array}{*{20}{c}}
{\frac{{ - {{\bf{v}}_{l - 1}}^T{{\bf{A}}_{l - 1}}^{ - 1}}}{{\sqrt {{\nu _{2l}} - {{\bf{v}}_{l - 1}}^T{{\bf{A}}_{l - 1}}^{ - 1}{{\bf{v}}_{l - 1}}} }}}&{\frac{1}{{\sqrt {{\nu _{2l}} - {{\bf{v}}_{l - 1}}^T{{\bf{A}}_{l - 1}}^{ - 1}{{\bf{v}}_{l - 1}}} }}}
\end{array}} \right]^{\rm{T}}}.
\end{equation}
In (\ref{eqn:def_c_lk}), ${{\bf{v}}_l} = {\left[ {\begin{array}{*{20}{c}}
{{\nu _{l + 1}}}& \cdots &{{\nu _{2l + 1}}}
\end{array}} \right]^T}$, ${\bf c}_0=[1]$,  ${\bf A}_0=[1]$, and ${{\bf A}_l}^{-1}$ is given recursively as
\begin{multline}\label{eqn:mat_A_iter}
%{{\bf{A}}_l}^{ - 1} = \left[ {\begin{array}{*{20}{c}}
%{{{\bf{A}}_{l - 1}}^{ - 1} + \frac{{{{\bf{A}}_{l - 1}}^{ - 1}{{\bf{v}}_{l - 1}}{{\bf{v}}_{l - 1}}^{\rm T}{{\bf{A}}_{l - 1}}^{ - 1}}}{{{\nu _{2l}} - {{\bf{v}}_{l - 1}}^{\rm T}{{\bf{A}}_{l - 1}}^{ - 1}{{\bf{v}}_{l - 1}}}}}&{ - \frac{{{{\bf{A}}_{l - 1}}^{ - 1}{{\bf{v}}_{l - 1}}}}{{{\nu _{2l}} - {{\bf{v}}_{l - 1}}^{\rm T}{{\bf{A}}_{l - 1}}^{ - 1}{{\bf{v}}_{l - 1}}}}}\\
%{ - \frac{{{{\bf{v}}_{l - 1}}^{\rm T}{{\bf{A}}_{l - 1}}^{ - 1}}}{{{\nu _{2l}} - {{\bf{v}}_{l - 1}}^{\rm T}{{\bf{A}}_{l - 1}}^{ - 1}{{\bf{v}}_{l - 1}}}}}&{\frac{1}{{{\nu _{2l}} - {{\bf{v}}_{l - 1}}^{\rm T}{{\bf{A}}_{l - 1}}^{ - 1}{{\bf{v}}_{l - 1}}}}}
%\end{array}} \right].
{{\bf{A}}_l}^{ - 1} = \left[ {\begin{array}{*{20}{c}}
{{{\bf{A}}_{l - 1}}^{ - 1} + \frac{{{{\bf{A}}_{l - 1}}^{ - 1}{{\bf{v}}_{l - 1}}{{\bf{v}}_{l - 1}}^{\rm{T}}{{\bf{A}}_{l - 1}}^{ - 1}}}{{{\nu _{2l}} - {{\bf{v}}_{l - 1}}^{\rm{T}}{{\bf{A}}_{l - 1}}^{ - 1}{{\bf{v}}_{l - 1}}}}}\\
{ - \frac{{{{\bf{v}}_{l - 1}}^{\rm{T}}{{\bf{A}}_{l - 1}}^{ - 1}}}{{{\nu _{2l}} - {{\bf{v}}_{l - 1}}^{\rm{T}}{{\bf{A}}_{l - 1}}^{ - 1}{{\bf{v}}_{l - 1}}}}}
\end{array}} \right.\\
\left. {\begin{array}{*{20}{c}}
{ - \frac{{{{\bf{A}}_{l - 1}}^{ - 1}{{\bf{v}}_{l - 1}}}}{{{\nu _{2l}} - {{\bf{v}}_{l - 1}}^{\rm{T}}{{\bf{A}}_{l - 1}}^{ - 1}{{\bf{v}}_{l - 1}}}}}\\
{\frac{1}{{{\nu _{2l}} - {{\bf{v}}_{l - 1}}^{\rm{T}}{{\bf{A}}_{l - 1}}^{ - 1}{{\bf{v}}_{l - 1}}}}}
\end{array}} \right].
\end{multline}

\end{theorem}
\begin{proof}
Please see Appendix \ref{app:proof_theorem1}.
\end{proof}

\begin{remark}\label{rem:orth}As shown in Appendix \ref{app:proof_remark1}, the set $\left\{ {{\bf c}_l^{\rm T} {\bf y}_l,l \in {\mathbb{N}}} \right\}$ in (\ref{eqn:app_re_mom_fina}) satisfies the orthogonality as
\begin{equation}\label{eqn:orth_from_theor1}
\left\langle {{{\bf{c}}_l}^{\rm{T}}{{\bf{y}}_l},{{\bf{c}}_k}^{\rm{T}}{{\bf{y}}_k}} \right\rangle  = \left\{ {\begin{array}{*{20}{c}}
{0,}&{l \ne k};\\
{{1},}&{l = k},
\end{array}} \right.
\end{equation}
where $\left\langle {g\left( y \right),h\left( y \right)} \right\rangle$ defines an inner product on $2$-norm Lebesgue spaces $L^2(\mathbb{R},\mathcal{F},u )$\footnote{Herein, $(\mathbb{R},\mathcal{F},\mu )$ is a measure space, where $\mathcal{F}$ is $\sigma$-algebra over $\mathbb{R}$.} with respect to a measure $dF_b(y) = f_b(y)dy$, such that
\begin{equation}\label{eqn:def_inner_produ}
\left\langle {g\left( y \right),h\left( y \right)} \right\rangle  = \int_{ 0 }^\infty  {g\left( y \right)h\left( y \right)dF_b(y)}.
\end{equation}
In other words, the set $\left\{ {{\bf c}_l^{\rm T} {\bf y}_l,l \in {\mathbb{N}}} \right\}$ can be regarded as an orthonormal basis with respect to the measure $dF_b(y)$. Therefore, the approximated PDF ${\tilde f_N}\left( y \right)$ in (\ref{eqn:app_re_mom_fina}) is in fact a linear combination of the orthonormal basis and $\eta_l$ can be regarded as the coordinate of ${\tilde f_N}\left( y \right)$ with respect to the basis vector ${\bf c}_l^{\rm T} {\bf y}_l$. The expression in (\ref{eqn:app_re_mom_fina}) will facilitate the analysis of the convergence of ${\tilde f_N}\left( y \right)$ with respect to $N$ and enable an efficient selection of the approximation degree $N$, which will be discussed later.
\end{remark}

Based on Theorem \ref{theorem:recur_coeff}, the PDF ${\tilde f_N}\left( y \right)$ now can be derived. To proceed, the base density function $f_b(y)$ should be determined first. As shown in Theorem \ref{the:inverse_mom_based1}, the base density function $f_b(y)$ should be chosen to be nontrivial with inverse moments $\nu_k$ existing. Notice that the base density function $f_b(y)$ is equivalent to the approximated PDF $\tilde f_N(y)$ when the approximation degree is set as zero, i.e., $f_{Y_K^D}(y) \approx \tilde f_N(y) = f_b(y)$ when $N=0$. The base density function $f_b(y)$ should be chosen somewhat close to $f_{Y_K^D}(y)$ \cite{provost2005moment,provost2012orthogonal}. Since the logarithm of $Y_K^D$ can be written as a sum of RVs, i.e., $\ln {Y_{K}^D} =  \sum\nolimits_{l = 1}^K {\ln \left( {1 + {\gamma_{SD,l}}} \right)}$, $\ln Y_K^D $ can be roughly approximated as a Gaussian RV based on the central limit theorem when the number of transmissions is large. It is thus natural to choose $f_b(y)$ as the PDF of a Lognormal RV given by
\begin{equation}\label{eqn_lognormal_den}
f{_b}\left( y \right) = \frac{1}{{y\sqrt {2\pi \sigma ^2} }}{e^{ - \frac{{{{\left( {\ln y - \mu } \right)}^2}}}{{2{\sigma }^2}}}} ,\, y \in \left( {0,\infty } \right),
\end{equation}
where $\mu $ and $\sigma ^2$ represent the mean and the covariance of $\ln(Y_K^D)$, respectively. As shown in Appendix \ref{app_m_sig}, the mean $\mu $ and the variance $\sigma ^2$ can be derived as
\begin{equation}\label{eqn_miu}
\mu=\sum\limits _{l=1}^{K} \underbrace{\frac{1}{{\Gamma\left(m\right)}}G_{2,3}^{3,1}\left({\left.{_{0,0,m}^{0,1}}\right|\frac{m}{{\Omega{'_{SD,l}}}}}\right)}_{\triangleq \mu_{l}},
\end{equation}
\begin{multline}\label{eqn_sigma}
  {\sigma ^2} \approx \sum\limits_{l = 1}^K {\left(\frac{1}{{\Gamma \left( m \right)}}{e^{\frac{m}{{\Omega {'_{SD,l}}}}}}\sum\limits_{p = 1}^{{N_Q}} {{w_p}{{\ln }^2}\left( {1 + \frac{{\Omega {'_{SD,l}}}}{m}{\zeta _p}} \right)}  - {\mu _l}^2\right)}  \\
   + 2\sum\limits_{i < j} {\left( \begin{array}{l}
\sum\limits_{{p_i},{p_j} \in \left[ {1,{N_Q}} \right]} {\frac{{\prod\limits_{l = i,j} {{w_{{p_l}}}} \ln \left( {1 + \frac{{{{\Omega }'_{SD,l}}\left( {1 - \lambda _{SD,l}^2} \right)}}{m}{\zeta _{{p_l}}}} \right)}}{{\Gamma \left( m \right){{\left( {1 + \sum\limits_{l = i,j} {\frac{{\lambda _{SD,l}^2}}{{1 - {\lambda _{SD,l}}^2}}} } \right)}^m}}}} \\
 \times \Psi _2^{\left( 2 \right)}\left( {m;m,m;\varpi _{i,j}^i{\zeta _{{p_i}}},\varpi _{i,j}^j{\zeta _{{p_j}}}} \right) - {\mu _i}{\mu _j}
\end{array} \right)} ,
\end{multline}
where $G_{p,q}^{m,n}(\cdot)$ represents Meijer G-function \cite[9.301]{gradshteyn1965table} and ${\varpi ^l_{i,j}} = {\left( {1 + \mathop {\mathop \sum }\limits_{k = i,j} \frac{{{\lambda _{SD,k}}^2}}{{1 - {\lambda _{SD,k}}^2}}} \right)^{ - 1}}\frac{{{\lambda _{SD,l}}^2}}{{1 - {\lambda _{SD,l}}^2}}$.

Given the base density function $f{_b}(y)$ in (\ref{eqn_lognormal_den}), its $k$th inverse moment ${\nu {_k}}$ directly follows as
\begin{equation}\label{eqn_mom}
{\nu _k} = \int_0^\infty  {{y^{-k}}{f_b}\left( y \right)dy}  = {e^{\frac{{{k^2}{{\sigma }^2}}}{2} - k\mu }}.
\end{equation}
The exponential form of ${\nu {_k}}$ then enables the derivation of the element in the vector ${\bf{c}}{_l}$ (\ref{eqn:def_c_lk}) in a closed-form as
\begin{align}\label{eqn_orth_coeff}
  {c_{l,k}} = \frac{{{{\left( { - 1} \right)}^{l + k}}{{e^{k\mu }}{\varsigma ^{\frac{{k - l - 2{k^2}}}{2}}}}\sqrt {\prod\nolimits_{l - 1 \ge t \ge 0} {\left( {1 - {\varsigma^{t - l}}} \right)} } }}{{\prod\nolimits_{t = 0,t \ne k}^l {\left( {1 - {\varsigma^{ - \left| {t - k} \right|}}} \right)} }}, % \notag \\
%  = {\left( { - 1} \right)^{l + k}}\frac{{{e^{\frac{{{{\sigma} ^2}\left( {k - l - 2{k^2}} \right) + 2k\mu }}{2}}}\sqrt {\prod\nolimits_{l - 1 \ge t \ge 0} {\left( {1 - {\varsigma^{t - l}}} \right)} } }}{{\prod\nolimits_{t = 0,t \ne k}^l {\left( {1 - {\varsigma^{ - \left| {t - k} \right|}}} \right)} }},\quad k \in [0,l]
\end{align}
where  $\varsigma = {e^{{{\sigma} ^2}}}$, as proved in Appendix \ref{app_coeff}.

Now putting (\ref{eqn_inv_mom_def}), (\ref{eta}), (\ref{eqn_lognormal_den}) and (\ref{eqn_orth_coeff}) into (\ref{eqn:app_re_mom_fina}), the approximated PDF ${{\tilde f}_N}\left( y \right)$ can be derived. Specifically, it can be written as
\begin{align}\label{eqn_f_N_Y}
{{\tilde f}_N}\left( y \right) &= {f_b}(y)\sum\limits_{l = 0}^N {{\eta _l}\sum\limits_{k = 0}^l {{c_{l,k}}{y^{ - k}}} } \notag \\
&= \sum\limits_{k = 0}^N {\sum\limits_{l = k}^N {{\eta _l}{c_{l,k}}{y^{ - k - 1}}\frac{1}{{\sqrt {2\pi {\sigma ^2}} }}{e^{ - \frac{{{{\left( {\ln y - \mu } \right)}^2}}}{{2{\sigma ^2}}}}}} }.
\end{align}
Accordingly, the approximated CDF ${{\tilde F}_N}\left( y \right)$ can be obtained as
\begin{align}\label{eqn_final_cdf_Y}
\tilde F_N\left( y \right) &= \sum\limits_{k = 0}^N {\sum\limits_{l = k}^N {{\eta _{ l}}{c _{l,k}\nu_k}} \Phi \left( {\frac{{\ln \left( y \right) + k\sigma {^2} - \mu }}{{\sigma }}} \right)} \notag \\
&= \sum\limits_{k = 0}^N {{\kappa _{ k}}\Phi \left( {\frac{{\ln \left( y \right) + k\sigma {^2} - \mu }}{{\sigma }}} \right)},
\end{align}
where $\Phi(\cdot)$ denotes the CDF of a standard normal RV and ${{\kappa _{ k}} = \sum\limits_{l = k}^N {{\eta _{ l}}} {c _{l,k} \nu_k}}$. Clearly, $\sum_{k = 0}^N {{\kappa _{k}}}  = 1$ since $F_N(\infty)  = 1$. It means that the approximated CDF $\tilde F_N\left( y \right)$ in fact is a weighted sum of the CDFs of Lognormal RVs.

By plugging (\ref{eqn_final_cdf_Y}) into (\ref{eqn_out_def_re_1}), the conditional outage probability ${P_{out}} \left(K|BC=K\right)$ can be derived as
\begin{align}\label{eqn_CDF_Z}
{P_{out}} \left(K|BC=K\right) &\approx \tilde F_N\left( 2^{\mathcal R} \right) \notag\\ &=\sum\limits_{k = 0}^N {{\kappa _k}\Phi \left( {\frac{{C{\cal R} + k{\sigma ^2} - \mu }}{\sigma }} \right)},
\end{align}
where $C=\ln 2$.

This inverse moment matching method can be applied to deriving the distributions of $Y_r^R$ and $Y_{K,r}^D$. Clearly in the derivation, it is essential to determine their inverse moments, and the mean and variance of their natural logarithms. Since $\left(\gamma_{SR,1},\cdots,\gamma_{SR,M}\right)$ follows a similar multivariate Gamma distribution with generalized correlation as $\left(\gamma_{SD,1},\cdots,\gamma_{SD,M}\right)$, the inverse moments of $Y_r^R$, and the mean and variance of $\ln Y_r^R$ can be derived respectively as (\ref{eqn_inv_mom_def}), (\ref{eqn_miu}) and (\ref{eqn_sigma}) but with the subscript $SD$ replaced as $SR$. Then the probability $\Pr \left(BC=r\right)$ can be obtained by putting the CDF of $Y_r^R$ into (\ref{eqn:out_rel_def_1}). With respect to $Y_{K,r}^D$, it can be written as a product of two independent RVs, i.e., $Y_{K,r}^D \triangleq Y_1 Y_2$, where $Y_1 = {\prod\nolimits_{l = 1}^r {\left( {1 + {\gamma_{SD,l}}} \right)} } $ and $Y_2 =  {\prod\nolimits_{l = r + 1}^K {\left( {1 + {\gamma_{RD,l}}} \right)} } $. Due to the independence of $Y_1$ and $Y_2$, the inverse moments of $Y_{K,r}^D$ are given by
\begin{equation}\label{eqn:Z_inver_mom}
{\rm {E}}\left( {Y_{K,r}^D}^{-k}\right) = {\rm {E}}\left({Y_1}^{-k}\right){\rm {E}}\left({Y_2}^{-k}\right) = {\tilde \alpha}_{1,k}{\tilde \alpha}_{2,k},
\end{equation}
where ${\tilde \alpha}_{1,k}$ and ${\tilde \alpha}_{2,k}$ denote the $k$th inverse moments of $Y_1$ and $Y_2$, respectively, which can be obtained similarly to (\ref{eqn_inv_mom_def}). Meanwhile, the mean and the variance of $\ln {Y_{K,r}^D}$ can be obtained as
\begin{equation}\label{eqn:Z_mean}
{\rm {E}}\left(\ln {Y_{K,r}^D}\right) = {\rm {E}}\left(\ln Y_1\right) + {\rm {E}}\left(\ln Y_2\right) = {{\tilde \mu _1}} + {{\tilde \mu _2}},
\end{equation}
\begin{equation}\label{eqn:Z_variance}
{\rm {Var}}\left(\ln {Y_{K,r}^D}\right) = {\rm {Var}}\left(\ln Y_1\right) + {\rm {Var}}\left(\ln Y_2\right) = {{\tilde \sigma} _1}^2 + {{\tilde \sigma} _2}^2,
\end{equation}
where $ ({{\tilde\mu _1}}, {{\tilde\sigma _1}{^2}}) $ and $( {{\tilde\mu _2}},{{\tilde\sigma _2}{^2}} )$ denote the mean and the variance of $\ln Y_1$ and $\ln Y_2$, respectively, which can be obtained similarly to (\ref{eqn_miu}) and (\ref{eqn_sigma}). Using the inverse moment matching method, the CDF of ${Y_{K,r}^D}$ can be finally derived as
\begin{equation}\label{eqn_cdf_I_K_k}
F_{Y_{K,r}^D}(y) \approx \sum\limits_{k = 0}^N {{\tilde \kappa _k}\Phi \left( {\frac{{{\ln(y)} + k{({\tilde \sigma_1} ^2 + {\tilde \sigma_2} ^2)} - (\tilde \mu_1 + \tilde\mu_2) }}{\sqrt{{\tilde\sigma_1} ^2 + {\tilde\sigma_2} ^2} }} \right)},
\end{equation}
where ${\tilde \kappa _0},{\tilde \kappa _1},\cdots,{\tilde \kappa _N}$ define the corresponding weightings of Lognormal CDFs. Then the conditional probability ${P_{out}} \left(K|BC=r\right)$ for $r < K$ can be obtained accordingly. Together with the probability of $\Pr \left(BC=r\right)$, the outage probability in (\ref{eqn_out_k_be_1}) directly follows.

%the Therefore, it turns to determine the moments of $ Y_r^R$, and the mean and the variance of $\ln Y_r^R$, which can be similarly derived as (\ref{eqn_inv_mom_def}), (\ref{eqn_miu}) and (\ref{eqn_sigma}), respectively. Hereby, $\Pr \left(BC=r\right)$ can be derived by putting the CDF of $Y_r^R$ into (\ref{eqn:out_rel_def_1}). While for the distribution $Y_{K,r}^D$, it will be derived later via extending inverse moment matching method.

\subsection{Selection of Approximation Degree}
In the inverse moment matching method, an truncation approximation is involved and the approximation degree $N$ should be properly chosen. To this end, the coordinates $\eta_l$ with respect to the orthonormal basis $\left\{ {{\bf c}_l^{\rm T} {\bf y}_l,l \in {\mathbb{N}}} \right\}$ should be analyzed. Recalling $\alpha_n \in (0,1]$ and putting (\ref{eqn_orth_coeff}) into (\ref{eta}), we have
\begin{align}\label{eqn:eta_N_upper}
\left| {{\eta _N}} \right| &\le \frac{{{\varsigma ^{\frac{{ - N}}{2}}}}}{{{\phi ^2}\left( \varsigma^{-1}  \right)}}\sum\limits_{k = 0}^N {{e^{{\sigma ^2}k\left( {\left( {\mu {\sigma ^{ - 2}} + \frac{1}{2}} \right) - k} \right)}}}  \notag \\
&\le \left( {N + 1} \right){\varsigma ^{\frac{{ - N}}{2}}}\frac{{{\varsigma^{{k_0}\left( {\left( {\mu {\sigma ^{ - 2}} + \frac{1}{2}} \right) - {k_0}} \right)}}}}{{{\phi ^2}\left( \varsigma^{-1}  \right)}}= A\left( {N + 1} \right){\varsigma ^{\frac{{ - N}}{2}}},
\end{align}
where $ {k_0} = \frac{{\mu {\sigma ^{ - 2}}}}{2} + \frac{1}{4}$, $\phi \left( q \right)$ denotes Euler function as $\phi \left( q \right) = \prod\limits_{k = 1}^\infty  {\left( {1 - {q^k}} \right)} $ and $A = \frac{{{\varsigma^{{k_0}\left( {\left( {\mu {\sigma ^{ - 2}} + \frac{1}{2}} \right) - {k_0}} \right)}}}}{{{\phi ^2}\left( \varsigma^{-1}  \right)}}$. Clearly, $\eta_N$ approaches to zero as $N$ tends to infinity, which justifies the truncation approximation.

To characterize the error between the PDF ${f_{Y_K^D}}\left( y \right)$ and its approximate ${{\tilde f}_N}\left( y \right)$ in (\ref{eqn:app_re_mom_fina}), a normalized error is generally defined as \cite{provost2012orthogonal}
\begin{equation}\label{eqn:error_pdf}
\epsilon_N\left( y \right) \triangleq \frac{{{f_{Y_K^D}}\left( y \right) - {{\tilde f}_N}\left( y \right)}}{{{f_b}(y)}} = \sum\limits_{l = N + 1}^\infty  {{\eta _l}{\bf{c}}_l^{\rm{T}}{{\bf{y}}_l}}.
\end{equation}
Accordingly, the normalized mean square error (NMSE) is defined as \cite{provost2012orthogonal}
\begin{align}\label{eqn:epsilon_sq_der}
{\left\| {{\epsilon_N}\left( y \right)} \right\|^2} &\triangleq {\left\langle {{\epsilon_N}\left( y \right),{\epsilon_N}\left( y \right)} \right\rangle } = \int_0^\infty  {{{\left( {\sum\limits_{l = N + 1}^\infty  {{\eta _l}{\bf{c}}_l^{\rm{T}}{{\bf{y}}_l}} } \right)}^2}d{F_b}(y)} \notag\\
 &= \sum\limits_{l = N + 1}^\infty  {\sum\limits_{k = N + 1}^\infty  {{\eta _l}{\eta _k}\left\langle {{\bf{c}}_l^{\rm{T}}{{\bf{y}}_l},{\bf{c}}_k^{\rm{T}}{{\bf{y}}_k}} \right\rangle } }
 = \sum\limits_{l = N + 1}^\infty  {{\eta _l}^2}.
\end{align}
Applying the upper bound (\ref{eqn:eta_N_upper}) into (\ref{eqn:epsilon_sq_der}), it follows that
\begin{align}\label{eqn:epsilon_sq_der_upp}
{\left\| {{\epsilon_N}\left( y \right)} \right\|^2} &\le \sum\limits_{l = N + 1}^\infty  {{{\left( {A\left( {l + 1} \right){\varsigma ^{\frac{{ - l}}{2}}}} \right)}^2}} \notag \\
& = {A^2}{\varsigma ^{ - N}}\sum\limits_{l = 1}^\infty  {{{\left( {l + N + 1} \right)}^2}{\varsigma ^{ - l}}} \notag \\
%& = {A^2}{\varsigma ^{ - N}}\left( {\sum\limits_{l = 1}^\infty  {{l^2}{\varsigma ^{ - l}}}  + {{\left( {N + 1} \right)}^2}\sum\limits_{l = 1}^\infty  {{\varsigma ^{ - l}}}  + 2\left( {N + 1} \right)\sum\limits_{l = 1}^\infty  {l{\varsigma ^{ - l}}} } \right) \notag \\
& = {A^2}{\varsigma ^{ - N}}\left( \begin{array}{l}
\frac{{{\varsigma ^{ - 2}} + {\varsigma ^{ - 1}}}}{{{{\left( {1 - {\varsigma ^{ - 1}}} \right)}^3}}} + {\left( {N + 1} \right)^2}\frac{{{\varsigma ^{ - 1}}}}{{1 - {\varsigma ^{ - 1}}}}\\
 + 2\left( {N + 1} \right)\frac{{{\varsigma ^{ - 1}}}}{{{{\left( {1 - {\varsigma ^{ - 1}}} \right)}^2}}}
\end{array} \right) \notag \\
& \le
 {A^2}{\left( {N + 1} \right)^2}{\varsigma ^{ - N}}\frac{{{\varsigma ^{ - 1}}}}{{1 - {\varsigma ^{ - 1}}}}\left( \begin{array}{l}
\frac{{1 + {\varsigma ^{ - 1}}}}{{{{\left( {1 - {\varsigma ^{ - 1}}} \right)}^2}}}\\
 + \frac{2}{{1 - {\varsigma ^{ - 1}}}} + 1
\end{array} \right),
\end{align}
where the last equality holds by using \cite[Eq.0.112, Eq.0.113, Eq.0.114]{gradshteyn1965table}. To guarantee the approximation accuracy, the NMSE should be limited to a small threshold $\varepsilon$, i.e., ${\left\| {{\epsilon_N}\left( y \right)} \right\|^2} \le \varepsilon$. To meet this error constraint and by defining $B = {A^2}\frac{{{\varsigma ^{ - 1}}}}{{1 - {\varsigma ^{ - 1}}}}\left( {\frac{{1 + {\varsigma ^{ - 1}}}}{{{{\left( {1 - {\varsigma ^{ - 1}}} \right)}^2}}} + \frac{{2{\varsigma ^{ - 1}}}}{{\left( {1 - {\varsigma ^{ - 1}}} \right)}} + 1} \right)$, the approximation degree $N$ should be chosen to satisfy
\begin{equation}\label{eqn:trun_order_obt}
B{\left( {N + 1} \right)^2}{\varsigma ^{ - N}} \le \varepsilon.
\end{equation}
It follows that
\begin{equation}\label{eqn:lambert_function_using}
\left( {N + 1} \right)\ln \sqrt {{\varsigma ^{ - 1}}} {e^{\left( {N + 1} \right)\ln \sqrt {{\varsigma ^{ - 1}}} }} \ge \sqrt {\frac{\varepsilon }{B}} \ln \sqrt {{\varsigma ^{ - 1}}} {e^{\ln \sqrt {{\varsigma ^{ - 1}}} }}.
\end{equation}
By the definition of Lambert W function, we have
\begin{equation}\label{eqn:using lambder_N}
\left( {N + 1} \right)\ln \sqrt {{\varsigma ^{ - 1}}} \le W_{-1}\left( {\sqrt {\frac{\varepsilon }{B}} \ln \sqrt {{\varsigma ^{ - 1}}} {e^{\ln \sqrt {{\varsigma ^{ - 1}}} }}} \right),
\end{equation}
where $W_{-1}(\cdot)$ denotes the lower branch of Lambert W function \cite[4.13]{olver2010nist}. It means that
\begin{align}\label{eqn:N_upp}
N &\ge \frac{{W_{-1}\left( {\sqrt {\frac{\varepsilon }{B}} \ln \sqrt {{\varsigma ^{ - 1}}} {e^{\ln \sqrt {{\varsigma ^{ - 1}}} }}} \right)}}{{\ln \sqrt {{\varsigma ^{ - 1}}} }} - 1 \notag \\
& = \frac{{W_{-1}\left( { - {\sigma ^2}/2\sqrt {\frac{\varepsilon }{B}} {e^{ - {\sigma ^2}/2}}} \right)}}{{ - {\sigma ^2}/2}} - 1 \triangleq {\bar N} .
\end{align}
%Here $\bar N$ can be regarded as an upper bound of the approximation degree.
For illustration, the approximation degree $\bar N$ under various NMSE constraints $\varepsilon$ is shown in Fig. \ref{fig_err}, by taking a system with parameters $m=6$, $\Omega {'_{SD,l}}=1$ and $\rho_{SD}^{l,k}=0.5$ as an example. Clearly, the approximation degree $\bar N$ decreases when higher NMSE is allowed.
%
%To meet this error constraint and by defining $B = {A^2}\frac{{{\varsigma ^{ - 1}}}}{{1 - {\varsigma ^{ - 1}}}}\left( {\frac{{1 + {\varsigma ^{ - 1}}}}{{{{\left( {1 - {\varsigma ^{ - 1}}} \right)}^2}}} + \frac{{2{\varsigma ^{ - 1}}}}{{\left( {1 - {\varsigma ^{ - 1}}} \right)}} + 1} \right)$, the approximation degree could be chosen no larger than an upper bound $\bar N$ as
%\begin{equation}\label{eqn:trun_order_obt}
%B{\left( {\bar N + 1} \right)^2}{\varsigma ^{ - \bar N}}= \varepsilon.
%\end{equation}
%It follows that
%\begin{equation}\label{eqn:lambert_function_using}
%\left( {\bar N + 1} \right)\ln \sqrt {{\varsigma ^{ - 1}}} {e^{\left( {\bar N + 1} \right)\ln \sqrt {{\varsigma ^{ - 1}}} }} = \sqrt {\frac{\varepsilon }{B}} \ln \sqrt {{\varsigma ^{ - 1}}} {e^{\ln \sqrt {{\varsigma ^{ - 1}}} }}.
%\end{equation}
%By the definition of Lambert W function, we have
%\begin{equation}\label{eqn:using lambder_N}
%\left( {\bar N + 1} \right)\ln \sqrt {{\varsigma ^{ - 1}}} = W_{-1}\left( {\sqrt {\frac{\varepsilon }{B}} \ln \sqrt {{\varsigma ^{ - 1}}} {e^{\ln \sqrt {{\varsigma ^{ - 1}}} }}} \right),
%\end{equation}
%where $W_{-1}(\cdot)$ denotes the lower branch of Lambert W function \cite[4.13]{olver2010nist}. It means that
%\begin{align}\label{eqn:N_upp}
%N \le \bar N &= \frac{{W_{-1}\left( {\sqrt {\frac{\varepsilon }{B}} \ln \sqrt {{\varsigma ^{ - 1}}} {e^{\ln \sqrt {{\varsigma ^{ - 1}}} }}} \right)}}{{\ln \sqrt {{\varsigma ^{ - 1}}} }} - 1 \notag \\
%& = \frac{{W_{-1}\left( { - {\sigma ^2}/2\sqrt {\frac{\varepsilon }{B}} {e^{ - {\sigma ^2}/2}}} \right)}}{{ - {\sigma ^2}/2}} - 1.
%\end{align}

To verify the inverse moment matching method, ${P_{out}} \left(K|BC=K\right)$ and ${P_{out}} \left(K|BC=r \right)$ with $r<K$ are plotted respectively in Fig. \ref{fig_cdf} and Fig. \ref{fig_poutKk}, taking a system with $m=6$, $\Omega {'_{SD,l}}=1$, $\Omega {'_{RD,l}}=2$, $r=2$ and $\rho_{SD}^{l,k}=\rho_{RD}^{l,k}=0.5$ as an example. Clearly, the gap between the analytical results and Monte Carlo simulation results significantly reduces with $N$. When $N=6$, the analytical result coincides well with simulation results, which validates its accuracy. In addition, the proposed inverse moment matching method performs better than the other two approaches, i.e., Lognormal approximation \cite{yang2014performance} and regular moment matching method \cite{provost2005moment,provost2012orthogonal}\footnote{{The regular moment matching method approximates the CDF of $Y_K^D$ as ${f_{Y_K^D}}(y) \approx {f_{\bar b}}(y)\sum\nolimits_{l = 0}^N {{\tilde \xi  _{N,l}}{y^l}}$ where ${f_{\bar b}}(y)$ is chosen as the PDF of a Lognormal RV with mean $-\mu$ and variance $\sigma^2$ while the coefficient ${\tilde \xi  _{N,l}}$ is determined by matching the first $N$ moments of $Y_K^D$. Since the moment generation function of $Y_K^D$ do not exist, this method cannot guarantee the uniqueness of the CDF, thus limiting the approximation accuracy.}}. Thus it justifies the effectiveness of the proposed method.
\begin{figure}
  \centering
  % Requires \usepackage{graphicx}
  \includegraphics[width=3in,height=2.4in]{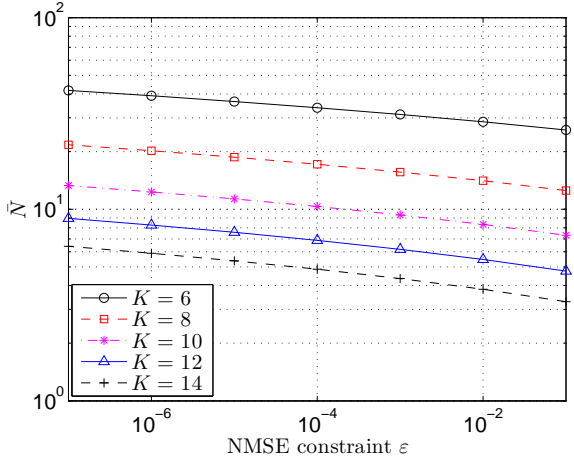}\\
  \caption{{The approximation degree $\bar N$ versus NMSE constraint $\varepsilon$}.}\label{fig_err}
\end{figure}
\begin{figure}
  \centering
  % Requires \usepackage{graphicx}
  \includegraphics[width=3.1in,height=2.4in]{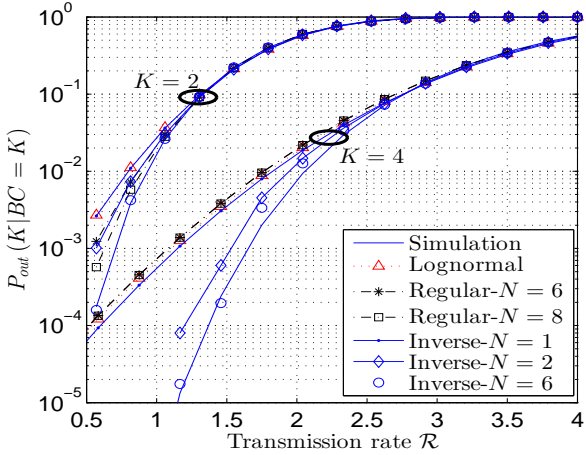}\\
  \caption{Effect of {approximation degree on ${P_{out}} \left(K|BC=K\right)$}.}\label{fig_cdf}
\end{figure}

\begin{figure}
  \centering
  % Requires \usepackage{graphicx}
  \includegraphics[width=3in,height=2.4in]{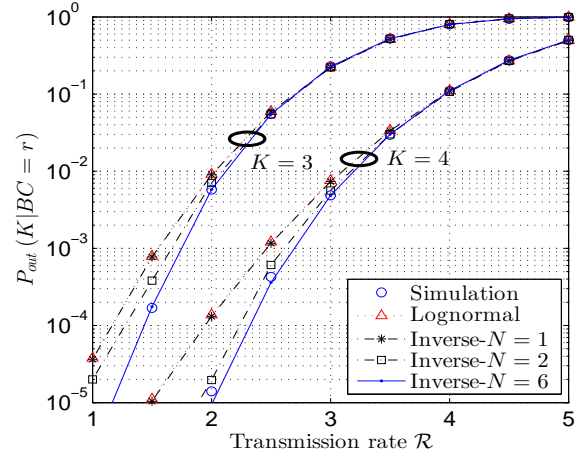}\\
  \caption{{Effect of approximation degree on ${P_{out}} \left(K|BC=r\right)$ with $ r = 2$}.}\label{fig_poutKk}
\end{figure}

\section{Diversity Order}\label{sec:div_ord}
To better understand the behavior of cooperative HARQ-IR schemes, another important performance metric (i.e., diversity order) is also analyzed here. Without loss of generality, the transmission SNR in each HARQ round is set equal, i.e., ${P_{S,l}}/{\mathfrak{N}_{SR,l}} = {P_{S,l}}/{\mathfrak{N}_{SD,l}} = {P_{R,l}}/{\mathfrak{N}_{RD,l}} = {\gamma _T}$ for $l \in [1,M]$. According to \cite{chelli2014performance,zheng2003diversity}, the diversity order $d$ is defined as
\begin{equation}\label{eqn:diver_def}
d =  - \mathop {\lim }\limits_{{\gamma _T} \to \infty } \frac{{\log \left( {P_{out}\left( M \right)} \right)}}{{\log \left( {{\gamma _T}} \right)}}.
\end{equation}

Based on the definition of outage probability in (\ref{eqn_out_k_be_1}) and noticing that the probabilities $\Pr \left(BC=r\right)$ are non-negative with $\sum\nolimits_{r = 1}^M  \Pr \left(BC=r\right)=1$, we have
\begin{multline}\label{eqn:out_upp_low_bound}
\min \left\{ {P_{out}} \left(M|BC=r\right), r \in [1,M] \right\} \le {P_{out}}\left( M \right)  \\
\le \max \left\{ {P_{out}} \left(M|BC=r\right), r \in [1,M]  \right\}.
\end{multline}
%Particularly, $p_M$ can be taken as a special case of $p_{M,r}$ by extending its definition with $r \ge M$. In what follows, we prove that the diversity order with respect to ${p_{M,r}}$ is equal to $Mm$, i.e., $d_{M,r} = - \mathop {\lim }\limits_{{\gamma _T} \to \infty } {{\log \left( {p_{M,r}} \right)}}/{{\log \left( {{\gamma _T}} \right)}} = Mm$.
Meanwhile, from the definition of $Y_{K,r}^D \triangleq \prod\nolimits_{l = 1}^r {\left( {1 + {\gamma _{SD,l}}} \right)} \times \prod\nolimits_{l = r + 1}^K {\left( {1 + {\gamma _{RD,l}}} \right)} $ in (\ref{eqn_out_def_re_1}), we also have
\begin{equation}\label{eqn:inequ_jens}
\left( {1 + {\bar \gamma _{M,r}}} \right) \le {Y_{M,r}^D}  \le \left( {1 + {M^{ - 1}}{\bar \gamma _{M,r}}} \right)^M,
\end{equation}
where ${\bar \gamma _{M,r}}$ represents the sum of SNRs as ${\bar \gamma _{M,r}} = \sum\nolimits_{l = 1}^r {{\gamma _{SD,l}}}  + \sum\nolimits_{l = r + 1}^M {{\gamma _{RD,l}}} $, and the right inequality follows from the inequality of arithmetic and geometric means. Applying (\ref{eqn:inequ_jens}) into (\ref{eqn_out_def_re_1}), the conditional outage probability ${P_{out}} \left(M|BC=r\right)$ is found to be bounded as
\begin{multline}\label{eqn:out_gama_M_r}
{F_{{\bar \gamma _{M,r}}}}\left( {M\left( {{2^{\frac{{\cal R}}{M}}} - 1} \right)} \right) \le {P_{out}} \left(M|BC=r\right) \\
\le {F_{{\bar \gamma _{M,r}}}}\left( {{2^{\cal R}} - 1} \right),
\end{multline}
where ${F_{{\bar \gamma _{M,r}}}}\left( \cdot \right)$ denotes the CDF of $\bar \gamma _{M,r}$ and is given by the following theorem.
\begin{theorem}
\label{theorem:CDF} The CDF of ${\bar \gamma _{M,r}} = \sum\nolimits_{l = 1}^r {{\gamma _{SD,l}}}  + \sum\nolimits_{l = r + 1}^M {{\gamma _{RD,l}}} $ can be written as
\begin{align}\label{eqn:cdf_Y_def_in_meijer_G}
&{F_{{\bar \gamma _{M,r}}}}\left( y \right) = \frac{{{y^{Mm}}}}{{{\gamma _T}^{Mm}{{\left( {\det \left( {\bf{B}} \right)} \right)}^m}\Gamma \left( {Mm + 1} \right)}} \times \notag\\
& \Phi _2^{\left( M \right)}\left( {m, \cdots ,m;Mm + 1; - \frac{y}{{{\gamma _T}{\delta _1}}}, \cdots , - \frac{y}{{{\gamma _T}{\delta _M}}}} \right),
\end{align}
where $\Phi _2^{\left( M \right)}\left(  \cdot  \right)$ denotes the confluent Lauricella function \cite[Def. A.19]{mathai2009h}, $\{\delta _k\}_{k=1}^{M}$ are defined as the eigenvalues of the matrix $\bf B=FE$, $\bf F$ is an $M \times M$ diagonal matrix with diagonal entries as $\{{{\Omega _{SD,1}}/m, \cdots ,{\Omega _{SD,r}}/m,{\Omega _{RD,r + 1}}/m, \cdots ,{\Omega _{RD,M}}/m}\}$, and $\bf E$ is an $M \times M$ symmetric positive definite matrix given by (\ref{eqn:C_mat_def}), shown at the top of this page.
\begin{figure*}[!t]
%\normalsize
  \centering
  % Requires \usepackage{graphicx}
\begin{equation}\label{eqn:C_mat_def}
{\bf{E}} = \left[ {\begin{array}{*{20}{c}}
{\begin{array}{*{20}{c}}
1&{\sqrt {\rho _{SD}^{1,2}} }& \cdots &{\sqrt {\rho _{SD}^{1,r}} }\\
{\sqrt{\rho _{SD}^{2,1}}}&1& \cdots &{\sqrt {\rho _{SD}^{2,r}} }\\
 \vdots & \vdots & \ddots & \vdots \\
{\sqrt {\rho _{SD}^{r,1}} }&{\sqrt {\rho _{SD}^{r,2}} }& \cdots &1
\end{array}}&{\bf{0}}_{r\times (M-r)}\\
{\bf{0}}_{(M-r)\times r}&{\begin{array}{*{20}{c}}
1&{\sqrt {\rho _{RD}^{r + 1,r + 2}} }& \cdots &{\sqrt {\rho _{RD}^{r + 1,M}} }\\
{\sqrt {\rho _{RD}^{r + 2,r + 1}} }&1& \cdots &{\sqrt {\rho _{RD}^{r + 2,M}} }\\
 \vdots & \vdots & \ddots & \vdots \\
{\sqrt {\rho _{RD}^{M,r + 1}} }&{\sqrt {\rho _{RD}^{M,r + 2}} }& \cdots &1
\end{array}}
\end{array}} \right], \quad 0 \le \rho_{SR}^{k,l},\rho_{RD}^{k,l} <1.
\end{equation}
\hrulefill
%\vspace*{4pt}
\end{figure*}
\end{theorem}
\begin{proof}
Under time-correlated fading channels, the SNRs in multiple HARQ rounds corresponding to one link are correlated, i.e., $\{{\gamma _{SD,1}},\cdots,{\gamma _{SD,r}}\}$ are correlated and
$\{{\gamma _{RD,r+1}},\cdots,{\gamma _{RD,M}}\}$ are correlated. The moment generating functions (MGFs) of the sum of correlated RVs, i.e., $\sum\nolimits_{l = 1}^r {{\gamma _{SD,l}}} $ and $\sum\nolimits_{l = r + 1}^M {{\gamma _{RD,l}}}$, can be derived as \cite[Eq. 8]{alouini2001sum}. Since $\sum\nolimits_{l = 1}^r {{\gamma _{SD,l}}} $ and $\sum\nolimits_{l = r + 1}^M {{\gamma _{RD,l}}} $ are independent, the MGF of ${\bar \gamma _{M,r}}$ can be directly written as the product of MGFs of $\sum\nolimits_{l = 1}^r {{\gamma _{SD,l}}} $ and $\sum\nolimits_{l = r + 1}^M {{\gamma _{RD,l}}} $. Then by applying inverse Laplace transform into the MGF of ${\bar \gamma _{M,r}}$, the CDF of ${\bar \gamma _{M,r}}$ can be derived as (\ref{eqn:cdf_Y_def_in_meijer_G}).
%
%Similar to \cite[Theorem 2.1]{kalyani2012asymptotic}, the CDF of ${\gamma _{M,r}}$ can be derived by using moment-generating function (MGF). With the result \cite[Eq. 8]{alouini2001sum},
%The MGFs of $\sum\nolimits_{l = 1}^r {{\gamma _{SD,l}}} $ and $\sum\nolimits_{l = r + 1}^M {{\gamma _{RD,l}}} $ can be obtained. Since $\sum\nolimits_{l = 1}^r {{\gamma _{SD,l}}} $ and $\sum\nolimits_{l = r + 1}^M {{\gamma _{RD,l}}} $ are independent, the MGF of ${\gamma _{M,r}}$ can be obtained straightforward by convolution property of Laplace transform. Then by applying inverse Laplace transform, the CDF of ${\gamma _{M,r}}$ can be derived as (\ref{eqn:cdf_Y_def_in_meijer_G}).
\end{proof}

From (\ref{eqn:diver_def}), (\ref{eqn:out_upp_low_bound}) and (\ref{eqn:out_gama_M_r}), it follows that
{\begin{align}\label{eqn:diversity_bet1}
& - \mathop {\lim }\limits_{{\gamma _T} \to \infty } \frac{\max \left\{{\log \left( {{F_{\bar \gamma _{M,r}}}\left( {\iota} \right)} \right)}, r \in [1,M]\right\}}{{\log \left( {{\gamma _T}} \right)}} \le d  \le  \notag \\
& - \mathop {\lim }\limits_{{\gamma _T} \to \infty } \frac{\min \left \{{\log \left( {{F_{\bar \gamma _{M,r}}}\left( \psi \right)} \right)}, r \in [1,M]\right \}}{{\log \left( {{\gamma _T}} \right)}}.
\end{align}
where ${\iota  = {2^{\cal R}} - 1}$ and $\psi={M\left( {{2^{{M^{ - 1}}{\cal R}}} - 1} \right)}$.}

With (\ref{eqn:cdf_Y_def_in_meijer_G}), the first inequality in (\ref{eqn:diversity_bet1}) can be rewritten as
\begin{align}\label{eqn:limit_ineq}
&d \ge Mm - \notag\\
&\max \left\{ \begin{array}{l}
\mathop {\lim }\limits_{{\gamma _T} \to \infty } \frac{{\log \Phi _2^{\left( M \right)}\left( {m, \cdots ,m;Mm + 1; - \frac{\iota }{{{\gamma _T}{\delta _1}}}, \cdots , - \frac{\iota }{{{\gamma _T}{\delta _M}}}} \right)}}{{\log {\gamma _T}}},\\
r \in [1,M]
\end{array} \right\}.
\end{align}
By using the series representation of the confluent Lauricella function \cite{aalo2005another}, the limit of the confluent Lauricella function is reduced as
\begin{multline}\label{eqn:inf_gamm_conf}
\mathop {\lim }\limits_{{\gamma _T} \to \infty } \Phi _2^{\left( M \right)}\left( {m, \cdots ,m;Mm + 1; - \frac{\iota }{{{\gamma _T}{\delta _1}}}, \cdots , - \frac{\iota }{{{\gamma _T}{\delta _M}}}} \right)  \\
 = 1.
\end{multline}
Putting (\ref{eqn:inf_gamm_conf}) into (\ref{eqn:limit_ineq}) yields $d \ge Mm$. Similarly, the second inequality in (\ref{eqn:diversity_bet1}) can be derived as $d \le Mm$. The diversity order then directly follows as $d=Mm$. Roughly speaking, a Nakagami-$m$ fading channel can be regarded as a set of $m$ parallel independent Rayleigh fading channels. Hereby, for HARQ operating over Nakagami-$m$ fast fading channels, the maximum achievable diversity order equals to the number of independently faded paths that the transmit signal experiences, i.e., $d_{max} = Mm$ \cite{zheng2003diversity}.
%\textcolor[rgb]{1.00,0.00,0.00}{With $M$ HARQ transmissions over Nakagami-$m$ fast fading channels, the maximum achievable diversity order is $Mm$ [Ref]}.
Therefore, it can be concluded that under time-correlated fading channels with $0 \le \rho_{SD}^{k,l},\rho_{RD}^{k,l}  < 1$, a full diversity order of $Mm$ can be achieved by this cooperative HARQ-IR scheme. Notice that under quasi-static fading channels, i.e., ${\rho_{ab}^{k,l}}=1$, since no time diversity can be gained from HARQ retransmissions in one link, the diversity order reduces to $m$ when $M=1$ and $2m$ when $M>1$ \footnote{The factor {2} comes due to the exploration of spatial diversity from the source and the relay.}.

\section{Numerical Results and Discussions}\label{sec:sim}
The analytical results derived would facilitate performance evaluation and enable optimal design of cooperative HARQ-IR systems over time-correlated Nakagami-$m$ fading channels. For illustration, we take systems with parameters $2\Omega_{SD,l} = \Omega_{SR,l} = \Omega_{RD,l} =1$ and a constant correlation model \cite{alouini2001sum,aalo1995performance,chen2004distribution}, i.e., $\rho_{ab}^{l,k}=\rho$ for $1 \le l \ne k \le M$ as examples. Unless otherwise stated, the transmission rate and the transmission SNR are set as $\mathcal R=4 \rm bps/Hz$ and $\gamma_T = 10 \rm dB$.

\subsection{Outage Performance Evaluation}
As shown in Fig. \ref{fig_outage_pro}, the outage probability $P_{out}(M)$ is plotted for systems with fading order $m=6$ and two different time correlations, i.e., $\rho=0.1$, $0.6$. It is easily seen that the analytical results match well with simulation results. The increase of the number of transmissions $M$ significantly decreases the outage probability, which demonstrates the benefit of HARQ-IR protocol. Given the number of transmissions $M>1$, the curves of outage probability under two different correlations become parallel as the transmission SNR $\gamma_T$ becomes large. Noticing that the outage probability is plotted as logarithmic scale, it means that $\log P_{out}(M)$ decreases at the same speed with the increase of $\log \gamma_T$ no matter what the correlation is. Here for the case of $M=1$, the curves under two different correlations merge together since only one transmission is allowed and the relay plays no role in the transmission. Moreover, the curves become steeper with the increase of $M$. These results are consistent with our analysis in Section \ref{sec:div_ord}, that is, the diversity order of cooperative HARQ-IR systems is equal to $Mm$ which is irrelevant to the time correlation.
\begin{figure}
  \centering
  % Requires \usepackage{graphicx}
  \includegraphics[width=3in,height=2.4in]{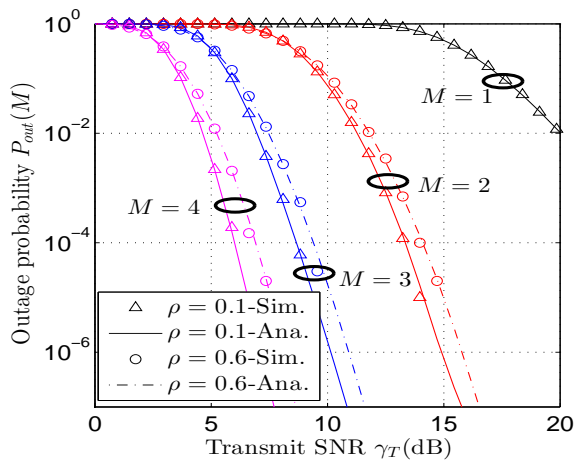}\\
  \caption{Verification of analytical results.}\label{fig_outage_pro}
\end{figure}

To further investigate the impact of time correlation on outage performance, Fig. \ref{fig:corr_out_ord} shows the outage probability $P_{out}(M)$ against time correlation $\rho$ for $M=3$. It is shown that the increase of time correlation would cause an outage performance degradation. For instance, the outage probability increases from $10^{-6}$ to $2*10^{-4}$ as $\rho$ increases from 0 to 1 given $m=6$. It therefore concludes that channel time correlation has a detrimental impact on outage performance.
\begin{figure}
  \centering
  % Requires \usepackage{graphicx}
  \includegraphics[width=3in,height=2.4in]{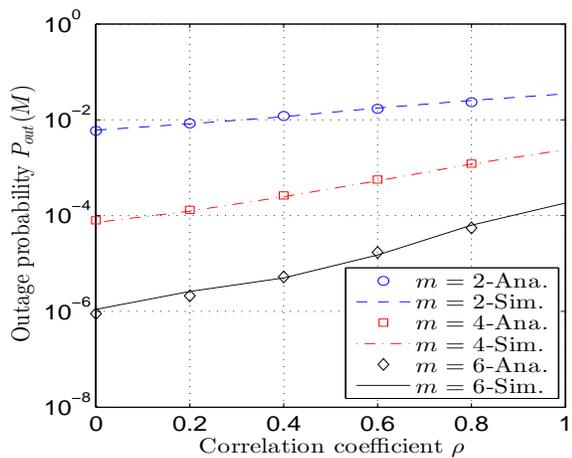}\\
  \caption{Impact of time correlation.
}\label{fig:corr_out_ord}
\end{figure}

Noticing that the fading order $m$ is an important parameter to characterize fading channels, the impact of fading order $m$ on outage performance is studied for $\rho=0.5$ in Fig. \ref{fig_outage_m}. Apparently, the increases of fading order would cause the decrease of outage probability. For example, given four transmissions $M=4$, the outage performance roughly achieves a 30dB gain when the fading order increases from 1 to 3. Thus we can conclude that the increase of fading order $m$ is beneficial to the outage performance, which has been particularly proved in Section \ref{sec:div_ord}, that is, the outage probability is directly proportional to ${\gamma_T}^{-Mm}$, i.e., $P_{out}(M) \propto {\gamma_T}^{-Mm}$.
\begin{figure}
  \centering
  % Requires \usepackage{graphicx}
  \includegraphics[width=3in,height=2.4in]{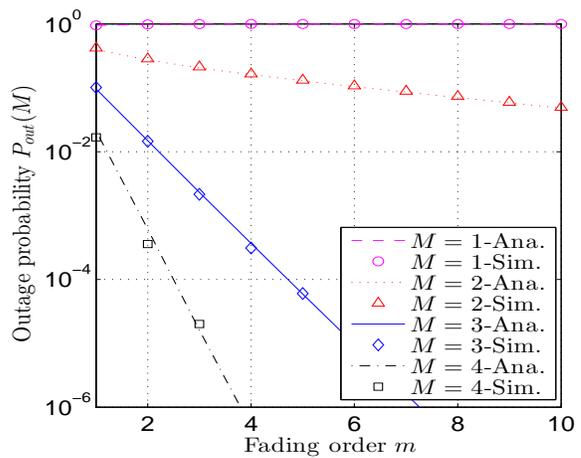}\\
  \caption{Impact of fading order.}\label{fig_outage_m}
\end{figure}

\subsection{Optimal Rate Selection}
Another widely concerned performance metric for HARQ-IR systems is long term average throughput (LTAT) and it is defined as \cite{caire2001throughput,zorzi1996use,zhao2005practical}
\begin{equation}\label{eqn_LTAT_def}
\bar{\mathcal  T}  = \frac{{{\cal R}\left( {1 - {P_{out}}\left( M \right)} \right)}}{{1 + \sum\nolimits_{l = 1}^{M - 1} {{P_{out}}\left( l \right)} }}.
\end{equation}
In practice, the HARQ-IR systems should usually be properly designed to achieve the maximum LTAT with guaranteed quality of service, e.g., a specifically low outage probability. Taking the design of the transmission rate as an example, the design problem can be formulated as
\begin{equation}\label{eqn_op}
\mathop {{\rm{max}}}\limits_{\cal R} \, \bar {\cal T}. \quad {\rm{s}}.{\rm{t}}.\, {P_{out}}\left( M \right) \le \vartheta,
\end{equation}
where $\vartheta$ specifies the outage constraint and denotes the maximum allowable outage probability. With our analytical results, the optimal rate and LTAT can be solved easily from  (\ref{eqn_op}) by using certain numerical tools. Given the maximum number of transmissions $M=4$, the optimal LTAT versus the outage constraint $\vartheta$ is shown in Fig. \ref{fig_avg_trans_rate_against_R1}. %Unlike the outage probability, there is no significant variation of the optimal LTAT under different $\rho$. Using certain numerical tools, the optimal LTAT $\bar {\mathcal T}_{opt}$ and the optimal transmission rate $\mathcal R_{opt}$ can be easily found as shown in Fig. \ref{fig_avg_trans_rate_against_R1}. For example, for the system with $\rho=0.6$ and $m=6$, the optimal transmission rate should be selected as $\mathcal R_{opt} = 7.4 \rm bps/Hz$ so that the maximum LTAT can be achieved as $\bar {\mathcal T}_{opt} = 2.17 \rm bps/Hz$.
It can be seen that the optimal LTAT $\bar {\mathcal T}_{opt}$ increases when the outage constraint $\vartheta$ is relaxed. However, no significant increase of $\bar {\mathcal T}_{opt}$ can be achieved when $\vartheta \ge 10^{-1}$.

%Moreover, since the outage probability ${P_{out}}\left( M \right) $ is an increasing function of $\mathcal R$, the outage constraint ${P_{out}}\left( M \right) \le \varepsilon$ in (\ref{eqn_op}) can be converted into $\mathcal R \le \bar {\mathcal R} $, where $\bar {\mathcal R}$ is the cut-off transmission rate. Using certain numerical tools or through exhaustively search in the region It is easily found that the optimization problem (\ref{eqn_op}) is non-convex and the optimal LTAT $\bar {\mathcal T}$ should be exhaustively searched in the region $\mathcal R \le \bar {\mathcal R} $. For example, for fixed values of $\rho=0.6$ and $m=6$, the optimal transmission rate should be selected as $\mathcal R = 7.4 \rm bps/Hz$.
%\begin{figure}
%  \centering
%  % Requires \usepackage{graphicx}
%  \includegraphics[width=5in]{./eps/ltat_rate1.eps}\\
%  \caption{LTAT $\bar{\mathcal T}$ against transmission rate $\mathcal R$ given outage constraint $\vartheta = 0.01$.}\label{fig_avg_trans_rate_against_R1_1}
%\end{figure}

\begin{figure}
  \centering
  % Requires \usepackage{graphicx}
  \includegraphics[width=3in,height=2.4in]{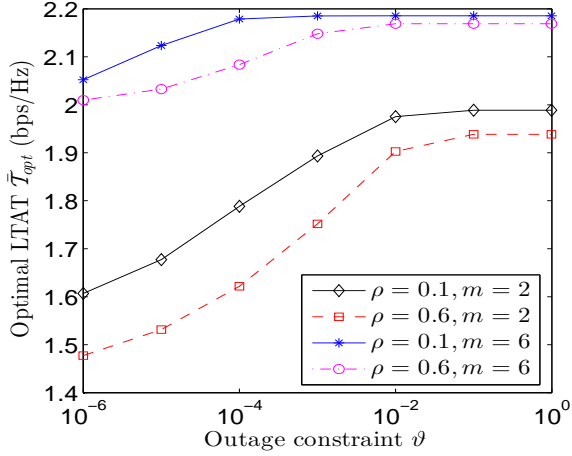}\\
  \caption{Optimal LTAT $\bar{\mathcal T}_{opt}$ against outage constraint $\vartheta$.}\label{fig_avg_trans_rate_against_R1}
\end{figure}
\section{Conclusions}
\label{sec_con}
In this paper, we have investigated the performance of cooperative HARQ-IR scheme operating over time-correlated Nakagami-$m$ fading channels. An efficient inverse moment matching method has been proposed to approximate the outage probability in closed-form as a weighted sum of multiple CDFs of Lognormal RVs. In addition, diversity order of cooperative HARQ-IR has been analyzed and it has been proved that full diversity can be achieved even under time correlated fading channels except quasi-static fading channels. The numerical results have demonstrated that high fading order and low time correlation are beneficial to the cooperative HARQ-IR scheme.

\appendices
\section{Nonexistence of MGFs corresponding to $Y_{K,r}^D$, $Y_K^D$, and $Y_r^R $}\label{app:proof_resu_1}
The nonexistence of MGFs with respect to $Y_{K,r}^D$, $Y_K^D$, and $Y_r^R $ can be proved by taking $Y_K^D$ as an example. Since $Y_K^D$ has finite moments of all order, its MGF can be written as \cite{billingsley2008probability}
\begin{equation}\label{eqn:mome_ge_fun}
\mathcal M_{Y_K^D}\left( s \right) = \sum\limits_{n = 0}^\infty  {\frac{{{{\beta _n}s^n}}}{{n!}}},
\end{equation}
where ${\beta _n}$ refers to the $n$th order moment of $Y_K^D$ given as
\begin{multline}\label{eqn:reg_mom_Y}
{\beta _n} = \int_{{\gamma _1} = 0}^\infty  { \cdots \int_{{\gamma _K} = 0}^\infty  {\prod\limits_{l = 1}^K {{{\left( {1 + {\gamma _l}} \right)}^n}} } }  \\
\times {f_{{\bf{\gamma }}_{SD}^{1:K}}}\left( {{\gamma _1}, \cdots ,{\gamma _K}} \right)d{\gamma _1} \cdots d{\gamma _K}.
\end{multline}
where ${{\boldsymbol{\gamma}}_{SD}^{1:K}} = (\gamma_{SD,1},\cdots,\gamma_{SD,K})$ with joint PDF denoted as ${f_{{{\boldsymbol{\gamma}}_{SD}^{1:K}}}}\left( {{\gamma_1}, \cdots ,{\gamma_K}} \right)$. By substituting (\ref{eqn:joint_gamma_l}) into (\ref{eqn:reg_mom_Y}) and making a change of variable ${z_l} = {{m\gamma_{SD,l}}}/{\left({\Omega {'_{SD,l}}\left( {1 - {\lambda _{SD,l}}^2} \right)}\right)}$, we have
\begin{multline}\label{eqn:beta_reg_mom_cal}
{\beta _n} = \int\nolimits_{t = 0}^\infty  {\frac{{{t^{m - 1}}}}{{{\Gamma ^{K + 1}}\left( m \right)}}} {{\rm{e}}^{ - \left( {1 + \sum\nolimits_{l = 1}^K {\frac{{{\lambda _{SD,l}}^2}}{{1 - {\lambda _{SD,l}}^2}}} } \right)t}} \\
 \times \prod\limits_{l = 1}^K {\int\nolimits_0^\infty  {z_l^{m - 1}} } {e^{ - {z_l}}}{\left( {1 + \frac{{\Omega {'_{SD,l}}\left( {1 - {\lambda _{SD,l}}^2} \right){z_l}}}{m}} \right)^n} \\
 \times {}_0{F_1}\left( {;m;\frac{{{\lambda _{SD,l}}^2t}}{{1 - {\lambda _{SD,l}}^2}}{z_l}} \right)d{z_l}dt.
\end{multline}
Since ${}_0F_1(;m;t) \ge 1$ for $t \ge 0$, the $n$th order moment ${\beta _n}$ is lower bounded by
\begin{multline}\label{eqn:reg_mom_fur_cal_ine}
{\beta _n} \ge \int\nolimits_{t = 0}^\infty  {\frac{{{t^{m - 1}}}}{{{\Gamma ^{K + 1}}\left( m \right)}}} {{\rm{e}}^{ - \left( {1 + \sum\limits_{l = 1}^K {\frac{{{\lambda _{SD,l}}^2}}{{1 - {\lambda _{SD,l}}^2}}} } \right)t}}dt  \\
\times \prod\limits_{l = 1}^K {{{\left( {\frac{{\Omega {'_{SD,l}}\left( {1 - {\lambda _{SD,l}}^2} \right)}}{m}} \right)}^n}\int\nolimits_0^\infty  {z_l^{m + n - 1}{e^{ - {z_l}}}d{z_l}} }.
\end{multline}
By using \cite[eq. 3.381.4]{gradshteyn1965table}, it follows that
\begin{multline}\label{eqn:reg_mom_fina_ineq}
{\beta _n} \ge {\left( {1 + \sum\limits_{l = 1}^K {\frac{{{\lambda _{SD,l}}^2}}{{1 - {\lambda _{SD,l}}^2}}} } \right)^{ - m}}\frac{{{\Gamma ^K}\left( {m + n} \right)}}{{{\Gamma ^K}\left( m \right)}}  \\
\times {\left( {\prod\limits_{l = 1}^K {\frac{{\Omega {'_{SD,l}}\left( {1 - {\lambda _{SD,l}}^2} \right)}}{m}} } \right)^n}.
\end{multline}
From (\ref{eqn:reg_mom_fina_ineq}), the following lower bound of ${{{\beta _n}{s^n}}}/{{n!}}$ holds
\begin{multline}\label{eqn:series_conve}
\frac{{{\beta _n}{s^n}}}{{n!}} \ge {\left( {1 + \sum\limits_{l = 1}^K {\frac{{{\lambda _{SD,l}}^2}}{{1 - {\lambda _{SD,l}}^2}}} } \right)^{ - m}}\frac{{{\Gamma ^K}\left( {m + n} \right)}}{{{\Gamma ^K}\left( m \right)n!}} \\
\times{\left( {s\prod\limits_{l = 1}^K {\frac{{\Omega {'_{SD,l}}\left( {1 - {\lambda _{SD,l}}^2} \right)}}{m}} } \right)^n} \triangleq {a_n},
\end{multline}
Since
\begin{align}\label{eqn:inftu_conv_ser}
\mathop {\lim }\limits_{n \to \infty } \frac{{{a_n}}}{{{a_{n - 1}}}} & = \mathop {\lim }\limits_{n \to \infty } \frac{{{{\left( {m + n - 1} \right)}^K}}}{n}s\prod\limits_{l = 1}^K {\frac{{\Omega {'_{SD,l}}\left( {1 - {\lambda _{SD,l}}^2} \right)}}{m}} \notag \\
& = \infty ,\,K > 1,
\end{align}
it is readily found that ${\beta_n s^n/n!} \ge a_n \to \infty$ as $n \to \infty$ if $K > 1$. Hereby, $\mathcal M_{Y_K^D}\left( s \right)$ in (\ref{eqn:mome_ge_fun}) diverges to infinity for any $s$. In other words, the MGF $\mathcal M_{Y_K^D}\left( s \right)$ does not exist. Similarly, we can prove that the MGFs of $Y_{K,r}^D$ and $Y_r^R $ do not exist either.

\section{Derivation of $\alpha_n$}
\label{app_inv_mom}
With the definition of ${Y_K^D \triangleq \prod\nolimits_{l = 1}^K {\left( {1 + {\gamma _{SD,l}}} \right)} }$, the $n$-th inverse moment of $Y_K^D$ can be written as
\begin{multline}\label{eqn_def_moments}
{\alpha_n}=
 \int\nolimits_{{\gamma_1} = 0}^\infty    \cdots \int\nolimits_{{\gamma_K} = 0}^\infty  \prod\limits_{l = 1}^K {{{\left( {1 + \gamma_l} \right)}^{ - n}}}  \\
 \times {f_{\bs\gamma_{SD}^{1:K}}}\left( {{\gamma_1}, \cdots ,{\gamma_K}} \right)d{\gamma_1} \cdots d{\gamma_K}   .
\end{multline}
Putting (\ref{eqn:joint_gamma_l}) into (\ref{eqn_def_moments}) and making a change of variable ${z_l} = {{m\gamma_{SD,l}}}/{\left({\Omega {'_{SD,l}}\left( {1 - {\lambda _{SD,l}}^2} \right)}\right)}$, it yields
\begin{multline}\label{eqn_def_moments_1}
{\alpha_n} = \int\nolimits_{t = 0}^\infty  {\frac{{{t^{m - 1}}}}{{{\Gamma ^{K + 1}}\left( m \right)}}} {{e}^{ - \left( {1 + \sum\limits_{l = 1}^K {\frac{{{\lambda _{SD,l}}^2}}{{1 - {\lambda _{SD,l}}^2}}} } \right)t}}\\
 \times \prod\limits_{l = 1}^K {\int\nolimits_0^\infty  {z_l^{m - 1}} } {e^{ - {z_l}}}{\left( {1 + \frac{{\Omega {'_{SD,l}}\left( {1 - {\lambda _{SD,l}}^2} \right){z_l}}}{m}} \right)^{ - n}} \\
 \times {}_0{F_1}\left( {;m;\frac{{{\lambda _{SD,l}}^2t}}{{1 - {\lambda _{SD,l}}^2}}{z_l}} \right)d{z_l}dt.
\end{multline}
By adopting Generalized Gaussian Quadrature \cite{rabinowitz1959tables,germund2008numerical}, the $n$-th inverse moment in (\ref{eqn_def_moments_1}) can be approximated as
\begin{multline}\label{eqn_def_moments_1_fur}
{\alpha_n}  \approx
  \int\nolimits_{t = 0}^\infty  {\frac{{{t^{m - 1}}}}{{{\Gamma ^{K + 1}}\left( m \right)}}} {{{e}}^{ - \left( {1 + \sum\limits_{k = 1}^K {\frac{{{\lambda _{SD,k}}^2}}{{1 - {\lambda _{SD,k}}^2}}} } \right)t}}\\
 \times \prod\limits_{l = 1}^K \sum\limits_{{p_l} = 1}^{{N_Q}} {{w_{{p_l}}}} {{\left( {1 + \frac{{\Omega {'_{SD,l}}\left( {1 - {\lambda _{SD,l}}^2} \right){\zeta _{{p_l}}}}}{m}} \right)}^{ - n}} \\
 \times {}_0{F_1}\left( {;m;\frac{{{\lambda _{SD,l}}^2}}{{1 - {\lambda _{SD,l}}^2}}{\zeta _{{p_l}}}t} \right) dt\\
 = \sum\limits_{{p_1}, \cdots ,{p_K} \in \left[ {1,{N_Q}} \right]} \frac{{\prod\limits_{l = 1}^K {{w_{{p_l}}}{{\left( {1 + \frac{{\Omega {'_{SD,l}}\left( {1 - {\lambda _{SD,l}}^2} \right){\zeta _{{p_l}}}}}{m}} \right)}^{ - n}}} }}{{{\Gamma ^{K + 1}}\left( m \right){{\left( {1 + \sum\limits_{l = 1}^K {\frac{{{\lambda _{SD,l}}^2}}{{1 - {\lambda _{SD,l}}^2}}} } \right)}^m}}} \\
 \times \int\nolimits_{t = 0}^\infty  {{t^{m - 1}}{e^{ - t}}\prod\limits_{l = 1}^K {{}_0{F_1}\left( {;m;{\varpi _l}{\zeta _{{p_l}}}t} \right)} dt} ,
\end{multline}
where $N_Q$ is the quadrature order, the weights $w_{p_l}$ and abscissas ${{\zeta _{p_l}}}$ for $N_Q$ up to $32$ are tabulated in \cite{rabinowitz1959tables}, and ${\varpi _l} = {{\frac{{{\lambda _{SD,l}}^2}}{{1 - {\lambda _{SD,l}}^2}}}({1 + \sum_{k = 1}^K {\frac{{{\lambda _{SD,k}}^2}}{{1 - {\lambda _{SD,k}}^2}}} }})^{-1}$. The approximation is valid for non-integer $m$ and can achieve a considerably high accuracy when $N_Q$ is sufficiently large \cite{rabinowitz1959tables,germund2008numerical}. Since the integral in (\ref{eqn_def_moments_1_fur}) can be derived as
\begin{align}\label{eqn_use_int}
&\int\nolimits_{t = 0}^\infty  {{t^{m - 1}}{e^{ - t}}\prod\limits_{l = 1}^K {_0{F_1}\left( {;m;{\varpi _l}{\zeta _{{p_l}}}t} \right)} } dt \notag \\
%& = \prod\limits_{l = 1}^K {\frac{{\Gamma \left( m \right)}}{{2\pi i}}\int\limits_{{\mathcal C_l}} {\frac{{\Gamma \left( {{s_l}} \right)}}{{\Gamma \left( {m - {s_l}} \right)}}{{\left( { - {\varpi _l}{\zeta _{{p_l}}}} \right)}^{ - {s_l}}}d{s_l}} } \int\limits_{t = 0}^\infty  {{t^{m - \sum\limits_{l = 1}^K {{s_l}}  - 1}}{e^{ - t}}} dt\notag\\
 & = {\left( {\frac{{\Gamma \left( m \right)}}{{2\pi i}}} \right)^K}\int\limits_{{{\cal C}_1}} { \cdots \int\limits_{{{\cal C}_K}} {\frac{{\Gamma \left( {m - \sum\limits_{l = 1}^K {{s_l}} } \right)\prod\limits_{l = 1}^K {\Gamma \left( {{s_l}} \right)} }}{{\prod\limits_{l = 1}^K {\Gamma \left( {m - {s_l}} \right){{\left( { - {\varpi _l}{\zeta _{{p_l}}}} \right)}^{{s_l}}}} }}} d{s_1} \cdots d{s_K}}  \notag\\
 &= {\Gamma ^K}\left( m \right)\Psi _2^{\left( K \right)}\left( {m;m, \cdots ,m;{\varpi _1}{\zeta _{{p_1}}}, \cdots ,{\varpi _K}{\zeta _{{p_K}}}} \right),
\end{align}
where $\Psi _2^{\left( K \right)}(;;)$ denotes the confluent form of Lauricella hypergeometric function \cite[Definition A.20]{mathai2009h} \cite{saigo1992some}, the $n$-th inverse moment (\ref{eqn_def_moments_1_fur}) is finally derived as (\ref{eqn_inv_mom_def}).

\section{Proof of Theorem \ref{the:inverse_mom_based1}}\label{app:inverse_mom_based}
It is clear from \cite{govindarajulu1962theory} that the $n$-th inverse moment of $Y_K^D$ is equivalent to the $n$-th moment of a RV $Z = 1/{Y_K^D}$, i.e.,
\begin{align}\label{eqn_inv_mom_exp}
\alpha_n &= \int\nolimits_{ 0 }^\infty  {{y^{ - n}}{f_{Y_K^D}}\left( y \right)dy}  = \int\nolimits_{ 0 }^\infty  {{z^n}{f_{Y_K^D}}\left( {\frac{1}{z}} \right)\frac{1}{{{z^2}}}dz} \notag \\
& = \int\nolimits_{ 0 }^\infty  {{z^n}{f_Z}\left( z \right)dz} ,
\end{align}
where $f_Z(z)$ is the PDF of $Z$ and satisfies that $f_Z(z)={{f_{Y_K^D}}\left( {{z^{-1}}} \right){{{z^{-2}}}}}$. With Property 1 and the result in \cite[pp. 176-177]{cramer1999mathematical}, the distribution of $Z$ can be uniquely determined by its moments $\alpha_n$. By using moment matching method, the PDF of $Z$ can be uniquely expressed as \cite{provost2005moment,provost2012orthogonal}
\begin{equation}\label{eqn_density_mom_match}
{f_Z}(z) = {f_a}(z)\sum\limits_{l = 0}^\infty {{\xi _l}{z^l}}.
\end{equation}
where ${f_a}(z)$ is a nontrivial base density function with moments ${\nu_{k}} = \int\nolimits_0^\infty  {{z^{k}}{f_a}\left( z \right)dz} $ existing, and $\xi _0,\xi _1\cdots$ denote the coefficients of the polynomial of $z$.
%can be explicitly determined by matching all the moments, such that
%\begin{equation}\label{eqn_match_mom}
%\alpha_n  = \int\limits_{ 0 }^\infty  {z^n {f_a}(z)\sum\limits_{l = 0}^\infty {{\xi _l}{z^l}}dz} ,\quad n = 0,1,\cdots.
%\end{equation}
Since $Y_K^D = 1/Z$, it follows from (\ref{eqn_density_mom_match}) that
\begin{align}\label{eqn:Y_trna_Z_app}
{f_{Y_K^D}}\left( y \right) &= {y^{ - 2}}{f_Z}\left( {{y^{ - 1}}} \right) = {y^{ - 2}}{f_a}\left( {{y^{ - 1}}} \right)\sum\limits_{l = 0}^\infty  {\xi _l{y^{ - l}}} \notag \\
 &= {f_b}\left( y \right)\sum\limits_{l = 0}^\infty  {\xi _l{y^{ - l}}},
\end{align}
where ${f_b}\left( y \right) \triangleq {y^{ - 2}}{f_a}\left( {{y^{ - 1}}} \right)$ and in fact is the inverse distribution with respect to ${f_a}\left( {{z}} \right)$. Therefore, ${f_b}\left( y \right)$ is a nontrivial function of $y$ and can be regarded as a base density function with inverse moments existing as $\int\nolimits_0^\infty  {{y^{-k}}{f_b}\left( y \right)dy}=\int\nolimits_0^\infty  {{z^{k}}{f_a}\left( z \right)dz}={\nu_{k}}$. According to Lemma \ref{lem:serie_con_unique}, since the PDF ${f_{Y_K^D}}\left( y \right)$ can be uniquely determined by matching all the inverse moments, the coefficients $\xi _0,\xi _1\cdots$ should satisfy
\begin{equation}\label{eqn:xi_app}
{\alpha _n} = \int\nolimits_0^\infty  {{y^{ - n}}{f_b}\left( y \right)\sum\limits_{l = 0}^\infty  {{\xi _l}{y^{ - l}}} dy}  = \sum\limits_{l = 0}^\infty  {{\xi _l}{\nu_{n + l}}} ,\, n = 0,1, \cdots.
\end{equation}

\section{Proof of Theorem \ref{theorem:recur_coeff}}\label{app:proof_theorem1}
In general, ${\bs\xi _N}$ can be written as
\begin{equation}\label{eqn:coeff_N_plus_1_up1}
{\bs\xi _N} = \left[ {\begin{array}{*{20}{c}}
{{\bs\xi _{N - 1}}}\\
0
\end{array}} \right] + {{\bf{e}}_N},
\end{equation}
where ${{\bf{e}}_N}$ characterizes the convergence of the coefficients $\bs\xi _l$ when the approximation degree $l$ is increased from $N-1$ to $N$. With (\ref{eqn:direc_inverse_matrix}) and (\ref{eqn:coeff_N_plus_1_up1}), ${{\bf{e}}_N}$ can be obtained as
\begin{equation}\label{eqn:coeff_N_plus_1_up1en}
{{\bf{e}}_N} = {{\bs{\xi }}_N} - \left[ {\begin{array}{*{20}{c}}
{{{\bs{\xi }}_{N - 1}}}\\
0
\end{array}} \right] = {{\bf{A}}_N}^{ - 1}\left( {{{\bs{\alpha }}_N} - {{\bf{A}}_N}\left[ {\begin{array}{*{20}{c}}
{{{\bs{\xi }}_{N - 1}}}\\
0
\end{array}} \right]} \right),
\end{equation}
where the second equality holds due to the invertibility of ${\bf A}_N$. From the definition in (\ref{eqn:matrix_U_mom_b}), ${{\bf{A}}_N}$ can be rewritten as
\begin{equation}\label{eqn:A_n_block}
{{\bf{A}}_N} = \left[ {\begin{array}{*{20}{c}}
{{{\bf{A}}_{N - 1}}}&{{{\bf{v}}_{N - 1}}}\\
{{{\bf{v}}_{N - 1}}^T}&{{\nu _{2N}}}
\end{array}} \right],
\end{equation}
where ${{\bf {v}}_N} = {\left[ {\begin{array}{*{20}{c}}
{{\nu _{N+1}}}&{{\nu _{N + 2}}}& \cdots &{{\nu _{2N + 1}}}
\end{array}} \right]^{\rm{T}}}$. By putting (\ref{eqn:A_n_block}) into (\ref{eqn:coeff_N_plus_1_up1en}) and using (\ref{eqn:coeff_matrx_form}), ${{\bf{e}}_N}$ can be further derived as
\begin{align}\label{eqn:delta_coeff_matr}
{{\bf{e}}_N} &= {{\bf{A}}_N}^{ - 1}\left( {{{\bs{\alpha }}_N} - \left[ {\begin{array}{*{20}{c}}
{{{\bf{A}}_{N - 1}}{{\bs{\xi }}_{N - 1}}}\\
{{{\bf{v}}_{N - 1}}^T{{\bs{\xi }}_{N - 1}}}
\end{array}} \right]} \right) \notag \\
%&= {{\bf{A}}_N}^{ - 1}\left( {\left[ {\begin{array}{*{20}{c}}
%{{{\bs{\alpha }}_{N - 1}}}\\
%{{\alpha _N}}
%\end{array}} \right] - \left[ {\begin{array}{*{20}{c}}
%{{{\bs{\alpha }}_{N - 1}}}\\
%{{{\bf{v}}_{N - 1}}^T{{\bf{A}}_{N - 1}}^{ - 1}{{\bs{\alpha }}_{N - 1}}}
%\end{array}} \right]} \right) \notag\\
 &= {{\bf{A}}_N}^{ - 1}\left[ {\begin{array}{*{20}{c}}
{\bf 0}_N\\
{{\alpha _N} - {{\bf{v}}_{N - 1}}^T{{\bf{A}}_{N - 1}}^{ - 1}{{\bs{\alpha }}_{N - 1}}}
\end{array}} \right],
\end{align}
where ${\bf 0}_{N}$ represents a null vector with length $N$.

By applying the inverse of a partitioned matrix \cite[5.16.b]{abadir2005matrix} on (\ref{eqn:A_n_block}), it yields
\begin{multline}\label{eqn:part_inver_matrix}
{{\bf{A}}_N}^{ - 1} =\\
 \left[ {\begin{array}{*{20}{c}}
{{{\left( {{{\bf{A}}_{N - 1}} - {{\bf{v}}_{N - 1}}{\nu _{2N}}^{ - 1}{{\bf{v}}_{N - 1}}^{\rm{T}}} \right)}^{ - 1}}}\\
{ - {{\left( {{\nu _{2N}} - {{\bf{v}}_{N - 1}}^{\rm{T}}{{\bf{A}}_{N - 1}}^{ - 1}{{\bf{v}}_{N - 1}}} \right)}^{ - 1}}{{\bf{v}}_{N - 1}}^{\rm{T}}{{\bf{A}}_{N - 1}}^{ - 1}}
\end{array}} \right.\\
\left. {\begin{array}{*{20}{c}}
{ - {{\left( {{{\bf{A}}_{N - 1}} - {{\bf{v}}_{N - 1}}{\nu _{2N}}^{ - 1}{{\bf{v}}_{N - 1}}^{\rm{T}}} \right)}^{ - 1}}{{\bf{v}}_{N - 1}}{\nu _{2N}}^{ - 1}}\\
{{{\left( {{\nu _{2N}} - {{\bf{v}}_{N - 1}}^{\rm{T}}{{\bf{A}}_{N - 1}}^{ - 1}{{\bf{v}}_{N - 1}}} \right)}^{ - 1}}}
\end{array}} \right].
\end{multline}
Using the matrix inversion lemma \cite[5.17]{abadir2005matrix}, we have
\begin{multline}\label{eqn:matrix_inverse_lemma}
{\left( {{{\bf{A}}_{N - 1}} - {{\bf{v}}_{N - 1}}{\nu _{2N}}^{ - 1}{{\bf{v}}_{N - 1}}^{\rm{T}}} \right)^{ - 1}} = \\
{{\bf{A}}_{N - 1}}^{ - 1} + \frac{{{{\bf{A}}_{N - 1}}^{ - 1}{{\bf{v}}_{N - 1}}{{\bf{v}}_{N - 1}}^{\rm{T}}{{\bf{A}}_{N - 1}}^{ - 1}}}{{{\nu _{2N}} - {{\bf{v}}_{N - 1}}^{\rm{T}}{{\bf{A}}_{N - 1}}^{ - 1}{{\bf{v}}_{N - 1}}}}.
\end{multline}
Then putting (\ref{eqn:matrix_inverse_lemma}) into (\ref{eqn:part_inver_matrix}) and after some manipulations, (\ref{eqn:part_inver_matrix}) can be eventually transformed into (\ref{eqn:mat_A_iter}).

Plugging (\ref{eqn:mat_A_iter}) into (\ref{eqn:delta_coeff_matr}), it produces
\begin{multline}\label{eqn:delta_xi_matr_fr}
{{\bf{e}}_N} = \frac{{{\alpha _N} - {{\bf{v}}_{N - 1}}^{\rm{T}}{{\bf{A}}_{N - 1}}^{ - 1}{\bs \alpha _{N - 1}}}}{{\sqrt {{\nu _{2N}} - {{\bf{v}}_{N - 1}}^{\rm{T}}{{\bf{A}}_{N - 1}}^{ - 1}{{\bf{v}}_{N - 1}}} }} \times \\
\underbrace {{\rm{ }}\left[ {\begin{array}{*{20}{c}}
{\frac{{ - {{\bf{A}}_{N - 1}}^{ - 1}{{\bf{v}}_{N - 1}}}}{{\sqrt {{\nu _{2N}} - {{\bf{v}}_{N - 1}}^{\rm{T}}{{\bf{A}}_{N - 1}}^{ - 1}{{\bf{v}}_{N - 1}}} }}}\\
{\frac{1}{{\sqrt {{\nu _{2N}} - {{\bf{v}}_{N - 1}}^{\rm{T}}{{\bf{A}}_{N - 1}}^{ - 1}{{\bf{v}}_{N - 1}}} }}}
\end{array}} \right]}_{{{\bf{c}}_N}}
%
%\left[ {\begin{array}{*{20}{c}}
%{ - {{\left( {{{\bf{A}}_{N - 1}} - {{\bf{v}}_{N - 1}}{\nu _{2N}}^{ - 1}{{\bf{v}}_{N - 1}}^{\rm{T}}} \right)}^{ - 1}}{{\bf{v}}_{N - 1}}{\nu _{2N}}^{ - 1}}\\
%{{{\left( {{\nu _{2N}} - {{\bf{v}}_{N - 1}}^{\rm{T}}{{\bf{A}}_{N - 1}}^{ - 1}{{\bf{v}}_{N - 1}}} \right)}^{ - 1}}}
%\end{array}} \right].
\end{multline}
 and (\ref{eqn:delta_xi_matr_fr}) can be further rewritten as
\begin{align}\label{eqn:delta_xi_final_exp}
&{{\bf{e}}_N}
%&\begin{array}{r}
%\underbrace {\left[ {\begin{array}{*{20}{c}}
%{\frac{{ - {{\bf{v}}_{N - 1}}^{\rm{T}}{{\bf{A}}_{N - 1}}^{ - 1}}}{{\sqrt {{\nu _{2N}} - {{\bf{v}}_{N - 1}}^{\rm{T}}{{\bf{A}}_{N - 1}}^{ - 1}{{\bf{v}}_{N - 1}}} }}}&{\frac{1}{{\sqrt {{\nu _{2N}} - {{\bf{v}}_{N - 1}}^{\rm{T}}{{\bf{A}}_{N - 1}}^{ - 1}{{\bf{v}}_{N - 1}}} }}}
%\end{array}} \right]}_{{{\bf{c}}_N}^{\rm{T}}}\\
%\times {\bs \alpha _N}\underbrace {\left[ {\begin{array}{*{20}{c}}
%{\frac{{ - {{\bf{A}}_{N - 1}}^{ - 1}{{\bf{v}}_{N - 1}}}}{{\sqrt {{\nu _{2N}} - {{\bf{v}}_{N - 1}}^{\rm{T}}{{\bf{A}}_{N - 1}}^{ - 1}{{\bf{v}}_{N - 1}}} }}}\\
%{\frac{1}{{\sqrt {{\nu _{2N}} - {{\bf{v}}_{N - 1}}^{\rm{T}}{{\bf{A}}_{N - 1}}^{ - 1}{{\bf{v}}_{N - 1}}} }}}
%\end{array}} \right]}_{{{\bf{c}}_N}}
%\end{array}\notag\\
 = {{\bf c}_N}^{\rm T} {{\bs{\alpha }}_N} {{\bf c}_N} = {\eta _N}{{\bf c}_N},
\end{align}
where $\eta _N = {{{\bf{c}}_N}^{\rm{T}}{{\bs{\alpha }}_N}} $, ${\bf c}_0=[1]$ and $\eta _0=1$. Substituting (\ref{eqn:delta_xi_final_exp}) into (\ref{eqn:coeff_N_plus_1_up1}), it follows that
\begin{multline}\label{eqn:eta_N_fina_rec}
{\bs \xi _N} = \left[ {\begin{array}{*{20}{c}}
{{\bs \xi _{N - 1}}}\\
0
\end{array}} \right] + { \eta _N}{{\bf{c}}_N} = \left[ {\begin{array}{*{20}{c}}
{{\bs \xi _{N - {\rm{2}}}}}\\
{{{\bf{0}}_{\rm{2}}}}
\end{array}} \right] \\
 +\left[ {\begin{array}{*{20}{c}}
{{ \eta _{N - 1}}{{\bf{c}}_{N - 1}}}\\
{\rm{0}}
\end{array}} \right]  + { \eta _N}{{\bf{c}}_N}
=\cdots = \sum\limits_{l = 0}^N {\left[ {\begin{array}{*{20}{c}}
{{ \eta _l}{{\bf{c}}_l}}\\
{{{\bf{0}}_{N - l}}}
\end{array}} \right]}.
\end{multline}
The proof then completes by substituting (\ref{eqn:eta_N_fina_rec}) into (\ref{eqn:ivner_bsed_fun_trun}).

\section{Proof of Remark \ref{rem:orth}}\label{app:proof_remark1}
With (\ref{eqn:delta_xi_final_exp}), $\left\langle {{{\bf{c}}_l}^{\rm{T}}{{\bf{y}}_l},{{\bf{c}}_k}^{\rm{T}}{{\bf{y}}_k}} \right\rangle $ can be written as
\begin{align}\label{eqn:inn_prod_rew}
&\left\langle {{{\bf{c}}_l}^{\rm{T}}{{\bf{y}}_l},{{\bf{c}}_k}^{\rm{T}}{{\bf{y}}_k}} \right\rangle \notag \\
&= \int_{ - \infty }^\infty  {{f_b}(y){{\bf{c}}_l}^{\rm{T}}{{\bf{y}}_l}{{\bf{y}}_k}^{\rm T}{{\bf{c}}_k}dy} = {\eta_l}^{-1}{\eta_k}^{-1}{{\bf{e}}_l}^{\rm{T}}{{\bf{A}}_{l,k}}{{\bf{e}}_k} \notag\\
&= {\eta_l}^{-1}{\eta_k}^{-1}\left[ {\begin{array}{*{20}{c}}
{\bf{0}}&{{\alpha _l} - {{\bf{v}}_{l - 1}}^{\rm T}{{\bf{A}}_{l - 1}}^{ - 1}{{\bs{\alpha }}_{l - 1}}}
\end{array}} \right] \notag\\
&\times {{\bf{A}}_l}^{ - 1}{{\bf{A}}_{l,k}}{{\bf{A}}_k}^{ - 1} \left[ {\begin{array}{*{20}{c}}
{\bf{0}}\\
{{\alpha _k} - {{\bf{v}}_{k - 1}}^{\rm T}{{\bf{A}}_{k - 1}}^{ - 1}{{\bs{\alpha }}_{k - 1}}}
\end{array}} \right],
\end{align}
where
\begin{equation}\label{eqn:def_A_lk}
{{\bf{A}}_{l,k}} = \left[ {\begin{array}{*{20}{c}}
{{\nu _0}}&{{\nu _1}}& \cdots &{{\nu _k}}\\
{{\nu _1}}&{{\nu _2}}& \cdots &{{\nu _{k + 1}}}\\
 \vdots & \vdots & \ddots & \vdots \\
{{\nu _l}}&{{\nu _{l + 1}}}& \cdots &{{\nu _{l + k}}}
\end{array}} \right],
\end{equation}
and the last step holds by using (\ref{eqn:delta_coeff_matr}).

For the case with $l=k$, ${{\bf{A}}_{l,k}} = {\bf A}_l$. It follows from (\ref{eqn:mat_A_iter}) and (\ref{eqn:inn_prod_rew}) that
\begin{align}\label{eqn:inner_pro_l_eql_k}
\left\langle {{{\bf{c}}_l}^{\rm{T}}{{\bf{y}}_l},{{\bf{c}}_l}^{\rm{T}}{{\bf{y}}_l}} \right\rangle   &= {\eta_l}^{-2}\frac{{{{\left( {{\alpha _l} - {{\bf{v}}_{l - 1}}^{\rm T}{{\bf{A}}_{l - 1}}^{ - 1}{{\bs{\alpha }}_{l - 1}}} \right)}^2}}}{{{\nu _{2l}} - {{\bf{v}}_{l - 1}}^{\rm T}{{\bf{A}}_{l - 1}}^{ - 1}{{\bf{v}}_{l - 1}}}} \notag\\
 &= {\eta_l}^{-2} {\left| {{{\bf{c}}_l}^{\rm{T}}{{\bs{\alpha }}_l}} \right|^2} = 1.
\end{align}
On the other hand, for the case with $l \ne k$, suppose that $l > k$ without loss of generality and then ${{\bf{A}}_l}^{ - 1}{{\bf{A}}_{l,k}}{{\bf{A}}_k}^{ - 1}$ can be written as
\begin{align}\label{eqn:inne_pro_der_matr_m}
&{{\bf{A}}_l}^{ - 1}{{\bf{A}}_{l,k}}{{\bf{A}}_k}^{ - 1} = \frac{1}{{\det \left( {{{\bf{A}}_l}} \right)\det \left( {{{\bf{A}}_k}} \right)}} \times \notag\\
&\begin{array}{r}
\left[ {\begin{array}{*{20}{c}}
{{\bf{A}}_l^{0,0}}&{{\bf{A}}_l^{1,0}}& \cdots &{{\bf{A}}_l^{l,0}}\\
{{\bf{A}}_l^{0,1}}&{{\bf{A}}_l^{1,1}}& \cdots &{{\bf{A}}_l^{l,1}}\\
 \vdots & \vdots & \ddots & \vdots \\
{{\bf{A}}_l^{0,l}}&{{\bf{A}}_l^{1,l}}& \cdots &{{\bf{A}}_l^{l,l}}
\end{array}} \right]\left[ {\begin{array}{*{20}{c}}
{{\nu _0}}&{{\nu _1}}& \cdots &{{\nu _k}}\\
{{\nu _1}}&{{\nu _2}}& \cdots &{{\nu _{k + 1}}}\\
 \vdots & \vdots & \ddots & \vdots \\
{{\nu _l}}&{{\nu _{l + 1}}}& \cdots &{{\nu _{l + k}}}
\end{array}} \right]\\
\times \left[ {\begin{array}{*{20}{c}}
{{\bf{A}}_k^{0,0}}&{{\bf{A}}_k^{1,0}}& \cdots &{{\bf{A}}_k^{k,0}}\\
{{\bf{A}}_k^{0,1}}&{{\bf{A}}_k^{1,1}}& \cdots &{{\bf{A}}_k^{k,1}}\\
 \vdots & \vdots & \ddots & \vdots \\
{{\bf{A}}_k^{0,k}}&{{\bf{A}}_k^{1,k}}& \cdots &{{\bf{A}}_k^{k,k}}
\end{array}} \right]
\end{array},
\end{align}
where ${{\bf{A}}_l^{i,j}}$ denotes the cofactor of the entry in the $i$-th row and $j$-th column of ${\bf A}_l$. Applying cofactor expansion of determinants into (\ref{eqn:inne_pro_der_matr_m}) yields
\begin{equation}\label{eqn:inne_pro_der_matr_m_fina}
\begin{array}{l}
{{\bf{A}}_l}^{ - 1}{{\bf{A}}_{l,k}}{{\bf{A}}_k}^{ - 1} = \frac{1}{{\det \left( {{{\bf{A}}_k}} \right)}}\left[ {\begin{array}{*{20}{c}}
{{\bf{A}}_k^{0,0}}&{{\bf{A}}_k^{1,0}}& \cdots &{{\bf{A}}_k^{k,0}}\\
 \vdots & \vdots & \ddots & \vdots \\
{{\bf{A}}_k^{0,k}}&{{\bf{A}}_k^{1,k}}& \cdots &{{\bf{A}}_k^{k,k}}\\
0&0&0&0\\
 \vdots & \vdots & \vdots & \vdots \\
0&0&0&0
\end{array}} \right]
\end{array}.
\end{equation}
Plugging (\ref{eqn:inne_pro_der_matr_m_fina}) into (\ref{eqn:inn_prod_rew}), we have $\left\langle {{{\bf{c}}_l}^{\rm{T}}{{\bf{y}}_l},{{\bf{c}}_k}^{\rm{T}}{{\bf{y}}_k}} \right\rangle = 0$ for $l \ne k$. Then the remark is proved.
%\begin{equation}\label{eqn:orth_k_large_k}
%\left[ {\begin{array}{*{20}{c}}
%{\bf{0}}&{{\alpha _l} - {{\bf{v}}_{l - 1}}^{\rm T}{{\bf{A}}_{l - 1}}^{ - 1}{{\bs{\alpha }}_{l - 1}}}
%\end{array}} \right]{{\bf{A}}_l}^{ - 1}{{\bf{A}}_{l,k}}{{\bf{A}}_k}^{ - 1}\left[ {\begin{array}{*{20}{c}}
%{\bf{0}}\\
%{{\alpha _k} - {{\bf{v}}_{k - 1}}^{\rm T}{{\bf{A}}_{k - 1}}^{ - 1}{{\bs{\alpha }}_{k - 1}}}
%\end{array}} \right] = 0.
%\end{equation}

\section{Derivation of $\mu$ and $\sigma ^2$}
\label{app_m_sig}
\subsection{Mean $\mu$}
The mean of RV $\ln(Y_K^D)$ is expressed as
\begin{align}\label{eqn_mean_def}
\mu  &= {\rm E}\left(\ln Y_K^D\right)=\int\nolimits_0^\infty  {\ln \left( y \right){f_{Y_K^D}}\left( y \right)dy} \notag \\
 &= \sum\limits_{l = 1}^K {\int\nolimits_{0}^\infty  {\ln \left( {1 + \gamma_l} \right){f_{{\gamma_{SD,l}}}}\left( {{\gamma_l}} \right)d{\gamma_l}}
 = \sum\limits_{l = 1}^K {{\mu _l}} },
\end{align}
where $\mu _l$ defines the expectation of $\ln(1+\gamma_l)$. Since $ \gamma_{SD,l} \sim \mathcal G(m,\Omega'_{SD,l}/m)$, $\mu_l$ is given by
\begin{multline}\label{eqn_mu_fur}
{\mu _l} = \frac{{{m^m}}}{{{{\left( {\Omega {'_{SD,l}}} \right)}^m}\Gamma \left( m \right)}} \\
\times {\int_0^\infty  {\ln \left( {1 + t} \right){t^{m - 1}}\exp \left( { - \frac{m}{{\Omega {'_{SD,l}}}}t} \right)dt} }.
\end{multline}
Applying Parseval equality of Meijer G-function \cite[Eq. 8.3.21]{debnath2010integral} into (\ref{eqn_mu_fur}) produces
\begin{align}\label{eqn_inte_f_text_1}
{\mu _l} &=  \frac{{{m^m}}}{{{{\left( {\Omega {'_{SD,l}}} \right)}^m}\Gamma \left( m \right)}} \int_0^\infty  {{t^{m - 1}}G_{2,2}^{1,2}\left( {\left. {_{1,0}^{1,1}} \right|t} \right)}\notag \\
&\times G_{0,1}^{1,0}\left( {\left. {_0^ - } \right|\frac{m}{{\Omega {'_{SD,l}}}}t} \right)dt
 %\int_0^\infty  {{t^{m - 1}}G_{2,2}^{1,2}\left( {\left. {_{1,0}^{1,1}} \right|t} \right)G_{0,1}^{1,0}\left( {\left. {_0^ - } \right|\frac{m}{{\Omega {'_{SD,l}}}}t} \right)dt} \notag\\
=\frac{1}{{\Gamma \left( m \right)}}G_{2,3}^{3,1}\left( {\left. {_{0,0,m}^{0,1}} \right|\frac{m}{{\Omega {'_{SD,l}}}}} \right).
\end{align}
Then by substituting (\ref{eqn_inte_f_text_1}) into (\ref{eqn_mean_def}), the mean $\mu$ is eventually derived as (\ref{eqn_miu}).
\subsection{Variance $\sigma ^2$}
According to the definition of variance of RV $\ln(Y_K^D)$, $\sigma ^2$ can be expressed as
\begin{align}\label{eqn_def_sig}
  {\sigma ^2} &= \sum\limits_{l = 1}^K {{\rm{Var}}\left[ {\ln \left( {1 + \gamma_{SD,l}} \right)} \right]}  + \notag \\
  &2\sum\limits_{1 \le i < j \le K} {{\rm Cov}\left( {\ln \left( {1 + \gamma_{SD,i}} \right),\ln \left( {1 + \gamma_{SD,j}} \right)} \right)},
\end{align}
where ${\mathop{\rm Var}} \left( {\ln \left( {1 + \gamma_{SD,l}} \right)} \right) = {\rm E}\left( {{{\ln }^2}\left( {1 + \gamma_{SD,l}} \right)} \right) - {\mu _l}^2$ and ${\rm Cov}\left( {\ln \left( {1 + \gamma_{SD,i}} \right),\ln \left( {1 + \gamma_{SD,j}} \right)} \right) ={\rm E}\left( {\ln \left( {1 + \gamma_{SD,i}} \right)\ln \left( {1 + \gamma_{SD,j}} \right)} \right) - {\mu _i}{\mu _j}$.
%\begin{equation}\label{eqn_var}
%  {\mathop{\rm Var}} \left( {\ln \left( {1 + \gamma_{SD,l}} \right)} \right) ={\rm E}\left( {{{\ln }^2}\left( {1 + \gamma_{SD,l}} \right)} \right) - {\mu _l}^2,
%\end{equation}
%\begin{multline}\label{eqn_cov}
%  {\rm Cov}\left( {\ln \left( {1 + \gamma_{SD,i}} \right),\ln \left( {1 + \gamma_{SD,j}} \right)} \right) \\
%  ={\rm E}\left( {\ln \left( {1 + \gamma_{SD,i}} \right)\ln \left( {1 + \gamma_{SD,j}} \right)} \right) - {\mu _i}{\mu _j}.
%\end{multline}

By making a change of variable $t={m}\gamma_{SD,l}/{\Omega'_{SD,l}}$, ${\rm E}\left( {{{\ln }^2}\left( {1 + \gamma_{SD,l}} \right)} \right)$ is given as
\begin{multline}\label{eqn_exp_log_xl}
{\rm E} \left(\ln^{2}\left(1+\gamma_{SD,l}\right)\right)=\frac{1}{{\Gamma \left( m \right)}}\exp \left( {\frac{m}{{{{\Omega '}_{SD,l}}}}} \right) \\
\times \int\nolimits_0^\infty  {{t^{m - 1}}} {e^{ - t}}{\ln ^2}\left( {1 + \frac{{{{\Omega '}_{SD,l}}}}{m}t} \right)dt.
\end{multline}
Using generalized Gaussian quadrature, it can be computed as
{\begin{multline}\label{eqn_exp_log_x1_fina_quad}
{\rm E}\left( {{{\ln }^2}\left( {1 + \gamma_{SD,l}} \right)} \right) \approx \\
\frac{1}{{\Gamma \left( m \right)}}{e^{\frac{m}{{\Omega {'_{SD,l}}}}}} \mathop \sum \limits_{p = 1}^{{N_Q}} {w_p}{\ln ^2}\left( {1 + \frac{{\Omega {'_{SD,l}}}}{m}{\zeta _p}} \right).
\end{multline}}
On the other hand, with (\ref{eqn:joint_gamma_l}), ${\rm E}\left( {\ln \left( {1 + \gamma_{SD,i}} \right)\ln \left( {1 + \gamma_{SD,j}} \right)} \right)$ can be written as
\begin{multline}\label{eqn_exp_log_log}
{\rm{E}}\left( {\ln \left( {1 + \gamma_{SD,i}} \right)\ln \left( {1 + \gamma_{SD,j}} \right)} \right) \\
  = \frac{1}{{{\Gamma ^3}\left( m \right)}}\int\nolimits_{t = 0}^\infty  {{t^{m - 1}}{e^{ - \left( {1 + \sum\limits_{l = i,j} {\frac{{{\lambda _{SD,l}}^2}}{{1 - {\lambda _{SD,l}}^2}}} } \right)t}}} \\
\times \prod\limits_{l = i,j} \int_0^\infty  {\begin{array}{*{20}{l}}
{y_l^{m - 1}{e^{ - {y_l}}}{}_0{F_1}\left( {;m;\frac{{{\lambda _{SD,l}}^2}}{{1 - {\lambda _{SD,l}}^2}}{y_l}t} \right)}
\end{array}} \\
\times \ln \left( {1 + \frac{{\Omega {'_{SD,l}}\left( {1 - {\lambda _{SD,l}}^2} \right)}}{m}{y_l}} \right)d{y_l} dt.
\end{multline}
Similar to (\ref{eqn_def_moments_1_fur}), (\ref{eqn_exp_log_log}) can be further derived by using the generalized Gaussian quadrature as
\begin{multline}\label{eqn_exp_log_log_fur}
{\rm E}\left( {\ln \left( {1 + \gamma_{SD,i}} \right)\ln \left( {1 + \gamma_{SD,j}} \right)} \right)
 \\
% \approx \frac{1}{{{\Gamma ^3}\left( m \right)}}\int\nolimits_{t = 0}^\infty  {{t^{m - 1}}{e^{ - \left( {1 + \sum\limits_{l = i,j} {\frac{{{\lambda _{SD,l}}^2}}{{1 - {\lambda _{SD,l}}^2}}} } \right)t}}} \\
% \times \prod\limits_{l = i,j} {\sum\limits_{{p_l} = 1}^{{N_Q}} {{w_{{p_l}}}{}_0{F_1}\left( {;m;\frac{{{\lambda _{SD,l}}^2}}{{1 - {\lambda _{SD,l}}^2}}{\zeta _{{p_l}}}t} \right)} } \\
%\times \ln \left( {1 + \frac{{\Omega {'_{SD,l}}\left( {1 - {\lambda _{SD,l}}^2} \right)}}{m}{\zeta _{{p_l}}}} \right)dt \\
 \approx \sum\limits_{{p_i},{p_j} \in \left[ {1,{N_Q}} \right]} {\frac{{\prod\limits_{l = i,j} {{w_{{p_l}}}\ln \left( {1 + \frac{{\Omega {'_{SD,l}}\left( {1 - {\lambda _{SD,l}}^2} \right)}}{m}{\zeta _{{p_l}}}} \right)} }}{{{\Gamma ^3}\left( m \right){{\left( {1 + \sum\limits_{l = i,j} {\frac{{{\lambda _{SD,l}}^2}}{{1 - {\lambda _{SD,l}}^2}}} } \right)}^m}}}} \\
 \times \int\nolimits_{t = 0}^\infty  {{t^{m - 1}}{e^{ - t}}} \prod\limits_{l = i,j} {{}_0{F_1}\left( {;m;\varpi _{i,j}^l{\zeta _{{p_l}}}t} \right)} dt,
\end{multline}
where ${\varpi ^l_{i,j}} = {\left( {1 + \mathop {\mathop \sum }\limits_{k = i,j} \frac{{{\lambda _{SD,k}}^2}}{{1 - {\lambda _{SD,k}}^2}}} \right)^{ - 1}}\frac{{{\lambda _{SD,l}}^2}}{{1 - {\lambda _{SD,l}}^2}},\, l = i,j$. Then putting (\ref{eqn_use_int}) into (\ref{eqn_exp_log_log_fur}) and together with (\ref{eqn_exp_log_x1_fina_quad}), the covariance $\sigma ^2$ can be  obtained as (\ref{eqn_sigma}).

\section{Derivation of $c_{l,k}$}\label{app_coeff}
Clearly from (\ref{eqn:def_c_lk}), to derive each element in ${\bf c}_l$, the terms ${\nu _{2l}} - {{\bf{v}}_{l - 1}}^{\rm T}{{\bf{A}}_{l - 1}}^{ - 1}{{\bf{v}}_{l - 1}}$ and ${\bf{A}}{_l}^{ - 1}{\bf{v}}{_l}$ should be determined first.

For the term ${\nu _{2l}} - {{\bf{v}}_{l - 1}}^{\rm T}{{\bf{A}}_{l - 1}}^{ - 1}{{\bf{v}}_{l - 1}}$, by using elementary transformation on the determinant of ${{\bf{A}}_l}^{ - 1}$ in (\ref{eqn:mat_A_iter}), we have
\begin{align}\label{eqn:deter_der_term_rel}
&\det \left( {{\bf{A}}{_l}^{ - 1}} \right) = \notag \\
&\det \left( {\left[ {\begin{array}{*{20}{c}}
{{\bf{A}}{_{l - 1}}^{ - 1}}&{\bf 0}\\
{ - \frac{{{\bf{v}}{_{l - 1}}^{\rm T}{\bf{A}}{_{l - 1}}^{ - 1}}}{{\nu {_{2l}} - {\bf{v}}{_{l - 1}}^{\rm T}{\bf{A}}{_{l - 1}}^{ - 1}{\bf{v}}{_{l - 1}}}}}&{\frac{1}{{\nu {_{2l}} - {\bf{v}}{_{l - 1}}^{\rm T}{\bf{A}}{_{l - 1}}^{ - 1}{\bf{v}}{_{l - 1}}}}}
\end{array}} \right]} \right).
\end{align}
It then follows that
\begin{equation}\label{eqn:term_con_deter1}
{\nu _{2l}} - {{\bf{v}}_{l - 1}}^{\rm T}{{\bf{A}}_{l - 1}}^{ - 1}{{\bf{v}}_{l - 1}} = \frac{{\det \left( {{{\bf{A}}_l}} \right)}}{{\det \left( {{{\bf{A}}_{l - 1}}} \right)}}.
\end{equation}
With respect to the term ${\bf{A}}{_l}^{ - 1}{\bf{v}}{_l}$, by expressing ${\bf{A}}{_l}^{ - 1}$ in terms of cofactors as
\begin{equation}\label{eqn:inverse_A}
{\bf{A}}{_l}^{ - 1} = \frac{1}{{\det \left( {{\bf{A}}{_l}} \right)}}\left[ {\begin{array}{*{20}{c}}
{{{A}}_l^{0,0}}&{{{A}}_l^{1,0}}& \cdots &{{{A}}_l^{N,0}}\\
{{{A}}_l^{0,1}}&{{{A}}_l^{1,1}}& \cdots &{{{A}}_l^{N,1}}\\
 \vdots & \vdots & \ddots & \vdots \\
{{{A}}_l^{0,N}}&{{{A}}_l^{1,N}}& \cdots &{{{A}}_l^{N,N}}
\end{array}} \right],
\end{equation}
${{\bf{A}}{_l}}^{ - 1}{\bf{v}}{_l}$ is rewritten as
\begin{equation}\label{eqn:term_con_deter}
{\bf{A}}{_l}^{ - 1}{\bf{v}}{_l} =  - \frac{1}{{\det \left( {{\bf{A}}{_l}} \right)}}{\left[ {\begin{array}{*{20}{c}}
{{{A}}_{l + 1}^{l + 1,0}}&{{{A}}_{l + 1}^{l + 1,1}}& \cdots &{{{A}}_{l + 1}^{l + 1,l}}
\end{array}} \right]^{\rm{T}}},
\end{equation}
where ${{{A}}_{l+1}^{i,j}}$ denotes the cofactor of the $(i,j)$-th entry of ${\bf {A}}_{l+1}$. Thus by plugging (\ref{eqn:term_con_deter1}) and (\ref{eqn:term_con_deter}) into (\ref{eqn:def_c_lk}), ${\bf{c}}{_l}$ is given as
\begin{equation}\label{eqn:c_l_as_cof}
{\bf{c}}{_l} = \frac{1}{{\sqrt {\det \left( {{{\bf{A}}_l}} \right)\det \left( {{\bf{A}}{_{l - 1}}} \right)} }}{\left[ {\begin{array}{*{20}{c}}
{{{A}}_l^{l,0}}&{{{A}}_l^{l,1}}& \cdots &{{{A}}_l^{l,l}}
\end{array}} \right]^{\rm{T}}}.
\end{equation}

By defining ${\bf{c}}{_l}=[c_{l,0}, c_{l,1}, \cdots, c_{l,l}]$ and using (\ref{eqn:c_l_as_cof}), the $k$th element $c_{l,k}$ can be expressed as
\begin{equation}\label{eqn:c_l_k_rew}
{c_{l,k}} = \frac{{{{A}}_l^{l,k}}}{{\sqrt {\det \left( {{\bf{A}}{_l}} \right)\det \left( {{\bf{A}}{_{l - 1}}} \right)} }} = \frac{{{{\left( { - 1} \right)}^{l + k}}{{U}}_l^{l,k}}}{{\sqrt {{{U}}{_l}{{U}}{_{l - 1}}} }},
\end{equation}
where ${{U}}{_l} = \det \left( {{\bf{A}}{_l}} \right)$, and ${{U}}_l^{l,k}$ denotes the corresponding minor of ${{A}}_l^{l,k}$. With the exponential form of $\nu_k$, both $U_l$ and $U_l^{l,k}$ can be simplified as Vandermonde determinants.
%By writing ${\bf{c}}{_l}=[c_{l,0}, c_{l,1}, \cdots, c_{l,l}]$ and using (\ref{eqn:c_l_as_cof}), $c_{l,k}$ can be expressed as
%\begin{equation}\label{eqn:c_l_k_rew}
%{c_{l,k}} = \frac{{{{A}}_l^{l,k}}}{{\sqrt {\det \left( {{\bf{A}}{_l}} \right)\det \left( {{\bf{A}}{_{l - 1}}} \right)} }} = \frac{{{{\left( { - 1} \right)}^{l + k}}{{U}}_l^{l,k}}}{{\sqrt {{{U}}{_l}{{U}}{_{l - 1}}} }}
%\end{equation}
%where ${{U}}{_l} = \det \left( {{\bf{A}}{_l}} \right)$, and ${{U}}_l^{l,k}$ denotes a minor corresponding to ${{A}}_l^{l,k}$, more precisely
%\begin{equation}\label{eqn:min_exp}
%{{U}}_l^{l,k} = \left| {\begin{array}{*{20}{c}}
%{\nu {_0}}& \cdots &{\nu {_{k - 1}}}&{\nu {_{k + 1}}}& \cdots &{\nu {_{l - 1}}}&{\nu {_l}}\\
%{\nu {_1}}& \cdots &{\nu {_k}}&{\nu {_{k + 2}}}& \cdots &{\nu {_l}}&{\nu {_{l + 1}}}\\
% \vdots & \vdots & \vdots & \vdots & \vdots & \vdots & \vdots \\
%{\nu {_{l - 1}}}& \cdots &{\nu {_{l + k}}}&{\nu {_{N + k + 1}}}& \cdots &{\nu {_{2l - 2}}}&{\nu {_{2l - 1}}}
%\end{array}} \right|
%\end{equation}
Specifically, ${{U}}{_l}$ can be written as
\begin{align}\label{eqn_U_n_form1}
{{U}}{_l} &= {{\left| {{\nu _{i + j}}} \right|}_{i,j \in \left[ {0,l} \right]}} = \left| {{e^{\frac{{{{\left( {i + j} \right)}^2}\sigma {^2}}}{2} - \left( {i + j} \right)\mu}}} \right|_{i,j \in \left[ {0,l} \right]} \notag\\
 %&= {e^{\sum\limits_{i = 0}^l {\left( {\frac{{{i^2}}}{2}\sigma {^2} - i\mu } \right)}  + \sum\limits_{j = 0}^l {\left( {\frac{{{j^2}}}{2}\sigma {^2} - j\mu } \right)} }}\left| {{e^{ij\sigma {^2}}}} \right|_{i,j \in \left[ {0,l} \right]} \notag\\
 & = {e^{ - l\left( {l + 1} \right)\mu }} {\varsigma ^{\frac{{l\left( {l + 1} \right)\left( {2l + 1} \right)}}{6}}}\left| {{\varsigma ^{ij}}} \right|_{i,j \in \left[ {0,l} \right]}
\end{align}
where $\varsigma = {e^{{{\sigma } ^2}}}$, and the notation $\left| {{\nu _{i + j}}} \right|_{i,j \in \left[ {0,l} \right]}$ represents the determinant of a matrix with $\nu _{i + j}$ as its $(i,j)$-th entry. Clearly, $\left| {{\varsigma^{ij}}} \right|_{i,j \in \left[ {0,l} \right]}$ is a Vandermonde determinant, henceforth ${{U}}{_l}$ can be obtained as
\begin{equation}\label{eqn_U_n_form2}
  {{{U}}{_l}} = {e^{ - l\left( {l + 1} \right)\mu }} {\varsigma ^{\frac{{l\left( {l + 1} \right)\left( {2l + 1} \right)}}{6}}} \prod\nolimits_{l \ge i > j \ge 0} {\left( {{\varsigma^i} - {\varsigma^j}} \right)}.
\end{equation}
Similarly, ${{U}}_l^{l,k}$ can also be simplified as a Vandermonde determinant given by
\begin{align}\label{eqn_U_n_k}
{U}_l^{l,k} &= {e^{\sum\limits_{i = 0}^{l - 1} {\left( {\frac{{{i^2}}}{2}{\sigma ^2} - i\mu } \right)}  + \sum\limits_{j = 0 \wedge j \ne k}^l {\left( {\frac{{{j^2}}}{2}{\sigma ^2} - j\mu } \right)} }}\notag \\
& \times {{\left| {{\varsigma^{ij}}} \right|}_{i \in \left[ {0,l - 1} \right],j \in \left[ {0,l} \right] \wedge j \ne k}} \notag\\
  &= \frac{{U_l}}{{{{\left( { - 1} \right)}^{l - k}}{\nu _l}{\nu _k}\prod\nolimits_{t = 0 \wedge t \ne k}^l {\left( {{\varsigma^k} - {\varsigma^t}} \right)} }}.
\end{align}

By substituting (\ref{eqn_U_n_form2}) and (\ref{eqn_U_n_k}) into (\ref{eqn:c_l_k_rew}), it yields
\begin{equation}\label{eqn:c_l_k_exp_der1}
\begin{array}{l}
{c_{l,k}}  = \sqrt {\frac{{U{_l}}}{{U{_{l - 1}}}}} \frac{1}{{\nu {_l}\nu {_k}\prod\nolimits_{t = 0 \wedge t \ne k}^l {\left( {{\varsigma^k} - {\varsigma^t}} \right)} }} = \frac{{\sqrt {\prod\nolimits_{l - 1 \ge t \ge 0} {\left( {{\varsigma^l} - {\varsigma^t}} \right)} } }}{{\nu {_k}\prod\nolimits_{t = 0 \wedge t \ne k}^l {\left( {{\varsigma^k} - {\varsigma^t}} \right)} }}
\end{array}.
\end{equation}
After some algebraic manipulations, (\ref{eqn:c_l_k_exp_der1}) is finally simplified as (\ref{eqn_orth_coeff}).
\bibliographystyle{IEEEtran}
\bibliography{manuscript_1}
\end{document}